\documentclass[ssy,preprint]{imsart}

\RequirePackage[OT1]{fontenc}
\usepackage{amsthm,amsmath,amssymb}
\usepackage{graphicx}
\usepackage{enumerate}
\RequirePackage[colorlinks,citecolor=blue,urlcolor=blue]{hyperref}

\usepackage{stackrel}
\usepackage{subfigure}
\usepackage{epsf}
\usepackage{color}
\usepackage{multirow,tabularx}
\usepackage{ifthen}
\usepackage{appendix}
\usepackage{tikz}
\usetikzlibrary{matrix}
\usepackage{epstopdf}

\arxiv{arXiv:0000.0000}

\theoremstyle{plain}
\newtheorem{theorem}{Theorem}[section]
\newtheorem{lemma}{Lemma}[section]
\newtheorem{proposition}{Proposition}[section]
\newtheorem{remark}{Remark}[section]
\newtheorem{definition}{Definition}[section]

\newcommand{\eqn}[1]{(\ref{#1})}

\newcommand{\brac}[1]{\left({#1}\right)}
\newcommand{\sbrac}[1]{\left[{#1}\right]}
\newcommand{\cbrac}[1]{\left\{{#1}\right\}}

\newcommand{\floor}[1]{\left\lfloor{#1}\right\rfloor}

\newcommand{\expect}[1]{\mathbb{E}\left[{#1}\right]}

\newcommand{\inprod}[2]{\langle{#1,#2} \rangle}
\newcommand{\ceil}[1]{\left\lceil{#1}\right\rceil}
\newcommand{\UN}{\bar{\mathcal{U}}^{(j)}_N}
\newcommand{\Ub}{\bar{\mathcal{U}}}
\newcommand{\U}{{\mathcal{U}}}
\newcommand{\norm}[1]{\|{#1}\|}
\newcommand{\abs}[1]{\left \vert {#1} \right \vert}
\newcommand{\mb}[1]{\mathbb{#1}}
\newcommand{\mf}[1]{\mathbf{#1}}
\newcommand{\bs}[1]{\boldsymbol{#1}}
\newcommand{\mv}[3]{{#1}^{(#3)}_{#2}}

\newcommand{\ksub}[3]{\floor{#1}_{#2 #3}}
\newcommand{\ktil}[3]{\ceil{#1}_{#2 #3}}

\startlocaldefs
\numberwithin{equation}{section}
\endlocaldefs

\begin{document}

\begin{frontmatter}

\title{Randomized Assignment of Jobs to Servers in Heterogeneous Clusters of Shared Servers for Low Delay}
\runtitle{Randomized job Assignment Schemes}

\begin{aug}

\author{\fnms{Arpan} \snm{Mukhopadhyay}\ead[label=e1]{arpan.mukhopadhyay@uwaterloo.ca}},
\address{Dept. of Electrical and Computer Engineering\\
University of Waterloo\\
Waterloo ON N2L 3G1,
Canada \\ \printead{e1}\\ 
\phantom{E-mail:\ }\printead*{e2}\\ \phantom{E-mail:\ }\printead*{e3}}
\author{\fnms{A.} \snm{Karthik}\ead[label=e2]{k4ananth@uwaterloo.ca}},
\and 
\author{\fnms{Ravi R.} \snm{Mazumdar}\corref{}\ead[label=e3]{mazum@uwaterloo.ca}}
\affiliation{University of Waterloo}

\runauthor{A. Mukhopadhyay, A. Karthik, and R. R. Mazumdar}

\end{aug}

\begin{abstract}
We consider the job assignment problem in 
a multi-server system consisting of $N$ parallel processor sharing servers,
categorized into $M$ ($\ll N$) different types
according to their processing capacity or speed. Jobs of 
random sizes arrive at the system
according to a Poisson process with rate $N \lambda$. Upon each arrival, 
a small number of servers from each type is sampled uniformly at random. The job is then
assigned to one of the sampled servers based on a selection rule. 
We propose two schemes, each corresponding to a specific selection rule
that aims at reducing the mean sojourn time of jobs in the system.

We first show that both methods achieve the maximal 
stability region. We then
analyze the system operating under the proposed schemes 
as $N \to \infty$ which corresponds to the mean field.
Our results show that asymptotic independence among servers holds even when
$M$ is finite and exchangeability holds only within servers of the same type. 
We further establish the existence and uniqueness of stationary solution of the mean field and
show that  the tail distribution of server occupancy decays doubly 
exponentially for each server type. When the estimates
of arrival rates are not available, the proposed schemes 
offer simpler alternatives to achieving lower mean sojourn time of jobs,
as shown by our numerical studies. 
\end{abstract}



\end{frontmatter}

\section{Introduction}
Consider a stream of jobs arriving at a multi-server system
consisting of a large number of parallel processor sharing servers.
The servers are categorized into different types or clusters according to their
processing capabilities. Each job,
upon arrival, is assigned to a server
where it completes its service and leaves the system.
The objective is to design job assignment schemes that reduce
the average sojourn, or response, time of jobs in the system. 

\subsection{Motivation}

The problem of job assignment is central in
multi-server
resource sharing systems that process delay sensitive web requests.
Examples include data centers and web server farms
running  applications such as 
online search, social networking etc.
In these systems, a small increase in 
the average response time of requests may cause
significant loss of revenue and users~\cite{Delay_study_google}.
Therefore, it is critical to reduce the average response time of jobs
in such systems.

Reduction in the average response time can be achieved by
assigning arrivals to less congested servers~\cite{Weber_optimal_JSQ,Gupta_Performance_2007,Winston_optimal_JSQ} 
in the system.
However, in today's systems, where the number of
front end servers is large, obtaining state information
of all the servers incurs a significant communication overhead.
For such systems, randomized job assignment schemes, in which
each assignment decision is made based on
comparing the states of a random subset of $d$ ($\geq 2$) servers,
are promising solutions. For systems with identical servers (homogeneous), 
such randomized schemes have been shown~\cite{Vvedenskaya_inftran_1996,Mitzenmacher_thesis,Graham} 
to result in a significant reduction
in mean response time of jobs as compared to state independent
schemes, in which job assignments are made independent of server states.
This implies that for large homogeneous systems, a small, randomly chosen subset of servers
is representative of the distribution of load in the overall system.
 
In this paper,
we consider heterogeneous systems where servers 
are grouped into different types or clusters, often geographically separated,
based on their capacities. 
Motivated by the aforementioned intuition arising
from the homogeneous case, 
we consider randomized job assignment schemes,
in which a small random subset of servers is sampled from each server type.
The least loaded servers of each type are then compared
based on the instantaneous processing rates they offer.
The job is then assigned to the server that provides the
highest processing rate. 
We consider processor sharing (PS) as the service discipline in this paper
since it closely approximates round-robin discipline with small granularity~\cite{Schassberger1984}
usually employed in server farms. Moreover, processor sharing discipline
has the desirable property of being insensitive to job length distribution type~\cite{Kelly_book}.

\subsection{Related literature}
 
Randomized job assignment schemes
have been
primarily studied in the literature for a system consisting of $N$ identical 
first come first serve (FCFS) servers, which is also referred
to as the supermarket model.
Most studies consider the so called shortest-queue-$d$ (SQ($d$))
scheme in which each job
is assigned to the shortest of $d$ randomly chosen queues.

For $d \geq 2$, ~\cite{Vvedenskaya_inftran_1996} showed, using the
theory of operator semigroups, that 
the equilibrium queue sizes decay doubly exponentially  
in the limit as the system size increases (as $N \to \infty$).
Mitzenmacher in~\cite{Mitzenmacher_thesis, Mitzenmacher_IEEE_2001}
derived the same result using an extension of Kurtz's theorem~\cite{Ethier_Kurtz_book}. 
In~\cite{Turner_choices_1998}, a coupling argument was used to show
that larger values of $d$ result in more even distribution of loads among the servers.
Chaoticity on path space (or asymptotic independence among queue length processes) 
was established in~\cite{Graham} using empirical
measures on the path space. Results of~\cite{Vvedenskaya_inftran_1996} 
were generalized to the case of open Jackson networks in~\cite{Martin_AAP_1999}.

Recently, in~\cite{Bramson_randomized_load_balancing},
the SQ($d$) scheme was analyzed under more general service disciplines 
and service time distributions. 
It was shown that in the case of FCFS discipline and power-law service time distribution, 
the equilibrium queue sizes decay doubly exponentially, 
exponentially, or just polynomially, depending on the power-law exponent and 
the number of choices, $d$. 
The stability of more general randomized schemes for non-idling service disciplines
was analyzed in~\cite{Bramson_stability_AAP}, which derived
a sufficient condition under which such networks
are stable. Asymptotic independence of servers in equilibrium
was proposed in~\cite{Bramson_asymp_indep} under local service disciplines
and general service time distributions. However, the result was proved only
for FCFS service discipline and service time distributions having decreasing 
hazard rate (DHR) functions.

The tradeoff between sampling cost of servers and
the expected sojourn time seen by a customer in the
supermarket model was studied under a game theoretic framework 
in~\cite{Hajek}. 
It was shown that for arrival rates within the stability region of the network, a symmetric Nash equilibrium
for identical customers exists in which each customer chooses a fixed number of queues to sample.

Recently, in~\cite{Mukhopadhyay_ITC_2014,Arpan_arxiv_2013}, 
the SQ($d$) scheme was considered 
for a system of parallel processor sharing servers with heterogeneous service rates.
It was shown that, in the heterogeneous setting, random sampling of $d$ servers 
from the entire system reduces the stability region.  
However, it can be recovered using the SQ($d$) scheme over a randomly chosen server type.

\subsection{Main results}

In this paper, we propose two new randomized schemes for job assignment
in the heterogeneous scenario. In both the schemes, upon arrival
of a job, a small number of servers of each type is randomly sampled. 
The sampled servers are then compared
based on their states and the arrival is assigned to the {\em best} server
among the chosen set of servers. The metric for choosing the best
server distinguishes the two schemes. 

This represents a scenario where a centralized dispatcher first requests information from each bank  or type of servers and 
then routes the job to the server that is going to give the lowest response time among the sampled servers. The number of servers sampled 
from a given type depends on the tradeoff between the sampling cost and the likely sojourn time as in the supermarket model in \cite{Hajek}.
We do not address the precise tradeoffs in this paper suffice to say that we assume that they could be different at each server type. 
We describe the precise mechanisms below.

In the first scheme, each arrival is assigned to the sampled server
with the least number of unfinished jobs. In the second,
each arrival is assigned to the sampled server offering the
maximum processing rate per unfinished job.
Note that, in the both the schemes, the sampled set contains
servers of all types. We show that such sampling
achieves the maximum possible stability region.

We analyze the performance of the proposed schemes in
the limit as the system size $N \to \infty$ using the mean field approach. 
Our analysis shows the following. 

\begin{itemize}
\item The stationary tail distribution of server occupancies
decay doubly exponentially in the limiting system.
We devise indirect methods to show this, since, 
unlike the homogeneous case, closed
form solutions of the stationary distribution cannot be obtained
in the heterogeneous scenario.


\item We establish the existence and uniqueness of the equilibrium
point of the mean field equations in the space of empirical
tail measures having finite first moment. Our proof, again,  differs
from the earlier works since closed form solutions cannot be obtained.
 
\item We show that propagation of chaos holds at each finite time
and also at the equilibrium.  In that, we generalize the earlier
results on propagation of chaos to systems where
exchangeability holds only among servers of the same type.
\end{itemize}
  
We also  numerically compare the proposed schemes with  existing schemes
for the heterogeneous case. It is observed that the proposed schemes
result in lower mean response time of jobs 
in scenarios where arrival rates cannot be estimated.

\subsection{Organization}
The rest of the paper is organized as follows. In Section~\ref{sec:model},
we describe the system model, the proposed job assignment schemes
and our notations.
We then analyze the proposed schemes in Sections~\ref{sec:stability},~\ref{sec:mean_field}, and~\ref{sec:computation}.
In Section~\ref{sec:numerics}, numerical results are presented that compare 
the schemes and determine
the accuracy of the theoretical results derived in the paper.
Finally, we conclude the paper in Section~\ref{sec:conclusion}
with a summary and a discussion on future work.

\section{Model and notations}
\label{sec:model}

We consider a multi-server system consisting of $N$ parallel 
processor sharing (PS) servers.
The capacity, $C$ (bits/sec), of a server is defined as the
time rate at which it processes a single job
present in it. If there are $q(t)$ jobs
present at a server of capacity $C$ at time $t$, then the instantaneous rate
at which each job is processed in the server is given by $C/q(t)$. 
Depending on their capacities, the servers in the system are categorized into $M$ ($\ll N$) types.
Define $\mathcal{J}=\cbrac{1,2,\ldots,M}$ to be the index set of server types.
The capacity of type $j$ servers is denoted by $C_j$, for $j \in \cal{J}$, 
and we assume, without loss of generality, that the server capacities are ordered in the following way:

\begin{equation}
C_1 \leq C_2 \leq \ldots \leq C_M.
\end{equation}
Further, for each $j \in \cal{J}$, we denote the proportion of type $j$ servers 
in the system by $\gamma_j$ ($0 \leq \gamma_j \leq 1$).
Clearly, $\sum_{j=1}^{M} \gamma_j=1$.

\begin{figure}
\centering
\includegraphics[width=.750\linewidth]{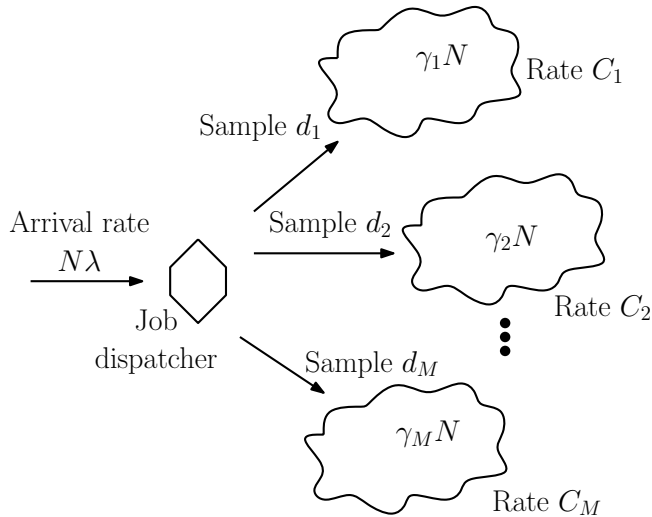}
\caption{System consisting of $N$ parallel 
processor sharing (PS) servers, categorized into $M$ types.
There are $\gamma_{j}N$ servers of type $j$, each of which has a capacity or rate $C_{j}$.
Arrivals occur according to a Poisson process with rate $N\lambda$. For each arrival,
the job dispatcher samples $d_{j}$ servers of type $j$ and routes the arrival to 
to one of the sampled servers.}
\label{fig:system_model}
\end{figure}

Jobs arrive at the system according to a Poisson process 
with rate $N\lambda$. Each job is of random  
length, independent and exponentially 
distributed with a finite mean $\frac{1}{\mu}$ (bits).\footnote{As discussed later,
our results do not depend on the
type of job length distribution due to the insensitivity of the processor
sharing discipline.}
The inter-arrival times and the job lengths are assumed
to be independent of each other.
Upon arrival, a job is assigned to one of the $N$ servers 
where  the job stays till the completion of its service after which it leaves the system. 
The model is illustrated in Figure~\ref{fig:system_model}.
We consider the following two job assignment schemes.

\subsection{Scheme~1}
In this scheme, upon arrival of a job, $d_j$ servers
of type $j$ are sampled uniformly at random from the set
of $N\gamma_j$ servers of type $j$, for each $j \in \cal{J}$. 
Note that this sampling is done at the cluster of type $j$ servers by a local 
router.

Let $\cbrac{\mv{q}{N}{j,1}, \mv{q}{N}{j,2}, \ldots, \mv{q}{N}{j,d_j}}$
denote the vector of occupancies of the $d_j$ sampled
servers of type $j$. For each type $j \in \cal{J}$,
a sampled server with index
$k_j$ is chosen for further comparison where
$k_j$ is given by

\begin{equation}
k_j=\operatorname{arg}\min_{1 \leq r \leq d_j} \cbrac{\mv{q}{N}{j,r}}.
\label{eq:first_level}
\end{equation}

In case of ties among sampled servers of type $j$, the index $k_j$ is chosen uniformly at random
from the tied servers of that type. The occupancy information of the server corresponding to $k_j$ is sent to the central dispatcher.

Using this information from each of the clusters $j \in {\mathcal J}$ 
the arriving job is assigned by the dispatcher to the 
type $i$ sampled server having index $k_i$ where

\begin{equation}
i=\operatorname{arg}\min_{1 \leq j \leq M} \cbrac{\mv{q}{N}{j,k_j}}.
\end{equation}

Ties across server types are broken by choosing the server type having the highest
capacity among the tied servers. Thus, in this scheme,
each arrival is assigned to the server having the least instantaneous occupancy
among the subset of randomly selected servers.

\subsection{Scheme~2}
As in Scheme~1,  upon arrival of a job, a random subset of
$d_j$ servers of type $j$ is chosen uniformly, for each $j \in \cal{J}$.  
Then from each type $j \in \cal{J}$,  a server with index $k_j$ is chosen 
according to~\eqref{eq:first_level}
for further comparison across different server types. 
The arriving job is finally assigned to
the type $i$ sampled server having index $k_i$ if

\begin{equation}
i=\operatorname{arg}\max_{1 \leq j \leq M} \cbrac{{C_j}/{\mv{q}{N}{j,k_j}}}.
\end{equation}

Note that the quantity $C_j/\mv{q}{N}{j,k_j}$ denotes the processing
rate per unfinished job at the sampled type $j$ server with index $k_j$. 
Thus, in this scheme, an arrival is assigned to the server
that provides the highest processing rate per job among the sampled
set of servers.
Ties are broken in the same way as described in Scheme~1. 

It is clear that Scheme~2 differs from Scheme~1 only in the criterion for server selection.
In Scheme~1, server selection is done based only on the instantaneous occupancies of
the sampled servers, whereas in Scheme~2 server capacities are also
taken into account in the selection criterion.  
Note that in the heterogeneous scenario a server with
higher occupancy can still provide a higher processing
rate than a server with lower occupancy. Therefore,
Scheme~2 provides a finer metric for server selection. 



\subsection{Notations}
 
We define the following real sequence spaces:

\begin{align}
\UN &= \{\cbrac{g_n}_{n \in \mathbb{Z}_{+}}: g_0=1, 
g_n \geq g_{n+1} \geq 0, N \gamma_j g_n \in \mathbb{N} \text{ }\forall n \in \mathbb{Z}_{+}\} \label{eq:UN},\\
\Ub &= \{\cbrac{g_n}_{n \in \mathbb{Z}_{+}}: g_0=1, 
g_n \geq g_{n+1} \geq 0 \text{ }\forall n \in \mathbb{Z}_{+}\}, \label{eq:Ub}\\
\U &= \{\cbrac{g_n}_{n \in \mathbb{Z}_{+}}: g_0=1, 
g_n \geq g_{n+1} \geq 0 \text{ }\forall n \in \mathbb{Z}_{+}, \sum_{n=0}^{\infty} g_n < \infty\}. \label{eq:U}
\end{align}
Let $\prod_{j \in \cal{J}} \UN$, $\Ub^M$, and $\U^M$ denote
the Cartesian products of $\UN$, $\Ub$, and $\U$, respectively, over $j \in \cal{J}$.
An element $\mf{u}=\left\{\mv{u}{n}{j}, j \in \mathcal{J}, n \in \mb{Z}_+\right\}$ belongs 
to $\prod_{j \in \mathcal{J}} \UN$, $\Ub^M$, or $\U^M$ if for each $j \in \cal{J}$,
the sequence $\cbrac{\mv{u}{n}{j}}_{n \in \mb{Z}_+}$ belongs to $\UN$,  $\Ub$, or $\U$, respectively.
For $\mf{u}, \mf{w} \in \Ub^M$ we define the following distance metric 

\begin{equation}
\norm{\mf{u}-\mf{w}}=\sup_{j \in \cal{J}} \sup_{n \in \mb{Z}_{+}} \abs{\frac{\mv{u}{n}{j}-\mv{w}{n}{j}}{n+1}}.
\label{eq:norm}
\end{equation}
It can be easily verified that under the metric defined in~\eqref{eq:norm},
the space $\Ub^M$ is compact (and hence complete and separable).
Further, for any $k \in \mathbb{Z}_{+}$ and $i, j \in \cal{J}$, we define
\begin{align}
\ksub{k}{i}{j} &= \floor{\frac{C_{j}}{C_{i}}k} + 1, \\
\ktil{k}{i}{j} &= \ceil{\frac{C_{j}}{C_{i}}k},
\end{align}
where $\floor{x}$ denotes the greatest integer not exceeding $x$ 
and $\ceil{x}$ denotes the smallest integer greater than or equal to $x$.

Let $(H, \mathcal{H}, \mu_H)$ be a measure space and 
$f: H \to \mb{R}$ be a 
$\mu_H$-integrable function. We define duality brackets as 
$\inprod{f}{\mu_H}=\int f d\mu_H$.  
We denote the weak convergence (convergence in distribution)
of a sequence of probability measures $P_n$ (random variables $X_n$) to
a probability measure $P$ (random variable $X$) by $P_n \Rightarrow P$ ($X_n \Rightarrow X$).

\section{Stability analysis}
\label{sec:stability}
In this section, we derive the sufficient condition for 
the system to have a finite expected number of jobs at all times
under Scheme~1 and Scheme~2.
In other words, we find the set of arrival rates for which
the Markov process describing the time evolution of the system 
is {\em positive Harris recurrent} or {\em stable}. We use the stability 
condition derived
in~\cite{Bramson_stability_AAP} for more general join-the-shortest-queue
(JSQ) networks.

\begin{theorem}
\label{thm:stability}
The system under consideration is stable under both Scheme~1 and Scheme~2 if

\begin{equation}
\lambda < \mu \sum_{j \in \cal{J}} \gamma_j C_j.
\label{eq:maximal_stability}
\end{equation}
\end{theorem}

\begin{proof}
Suppose that the $N$ servers in the system are indexed by the set $\mathcal{S}=\cbrac{1,2,\ldots,N}$.
For each job, we define a {\em selection set} to be the subset of $\sum_{j \in \cal{J}} d_j$ servers
sampled at its arrival. 
We denote by $p_A$ the probability that the subset $A \subseteq \mathcal{S}$ is chosen as the selection set
for an arrival. 
Note that $p_A$, $A \subseteq \cal{S}$, defines the
job assignment scheme used.
Under Scheme~1 and Scheme~2, the probability $p_A$ is non-zero only for subsets $A$ which contain $d_j$
servers of type $j$ for all $j \in \cal{J}$ and for each such a subset $A$, the probability $p_A$ is given by

\begin{equation}
p_A=\frac{1}{\prod_{j \in \mathcal{J}}\binom{N\gamma_j}{d_j}}.
\label{eq:select_rule}
\end{equation}
Now according to Corollary 1.1 of~\cite{Bramson_stability_AAP}, the system under consideration
is stable if it is {\em subcritical}, i.e., if it satisfies condition (1.2) of~\cite{Bramson_stability_AAP}.
Note that the additional conditions (1.11) and (1.12) of Corollary 1.1 of~\cite{Bramson_stability_AAP} 
are automatically satisfied since interarrival times are exponentially distributed. Applying
condition (1.2) of~\cite{Bramson_stability_AAP} to the system under consideration, we obtain 
the sufficient condition for stability of the system to be

\begin{equation}
\rho= \max_{B \subseteq \mathcal{S}} \cbrac{\brac{\mu \sum_{n \in B} C_{(n)}}^{-1} N \lambda \sum_{A \subseteq B} p_A} < 1,
\label{eq:subcriticality}
\end{equation} 
where $C_{(n)}$ denotes the capacity of the server with index $n$ in the set $\cal{S}$. Clearly,
for Scheme~1 and Scheme~2, the term within the braces in~\eqref{eq:subcriticality} is non-zero
only when the subset $B$ is composed of at least $d_j$ servers of type $j$ for all $j \in \cal{J}$. 
Let $B_j$ ($\geq d_j$) denote the 
number of type $j$ servers in $B$. Using~\eqref{eq:select_rule} and~\eqref{eq:subcriticality}
we now have

\begin{equation}
\rho=\max_{B \subseteq \mathcal{S}: N \gamma_j \geq B_j \geq d_j \forall \text{ }j \in \mathcal{J}} \cbrac{\frac{N\lambda}{\mu}\frac{1}{\sum_{j \in \cal{J}}B_j C_j} \prod_{j \in \mathcal{J}} \frac{\binom{B_j}{d_j}}{\binom{N\gamma_j}{d_j}}}.
\label{eq:subcriticality_simplified}
\end{equation}
It is easy to verify that that the function $\frac{\prod_{j \in \cal{J}} \binom{B_j}{d_j}}{\sum_{j \in \cal{J}} B_j C_j}$
is increasing with respect to $B_j$ for each $j \in \mathcal{J}$. Hence, the expression within the braces
in~\eqref{eq:subcriticality_simplified} is maximized when we set $B_j=N\gamma_j$. Hence, we have

\begin{equation}
\rho= \frac{N \lambda}{\mu} \frac{1}{N \sum_{j \in \cal{J}}\gamma_j C_j}= \frac{\lambda}{\mu\sum_{j \in \cal{J}} \gamma_j C_j}
\label{eq:subcriticality_simplified1}
\end{equation}
Therefore, from~\eqref{eq:subcriticality} and~\eqref{eq:subcriticality_simplified1} we conclude that the system
under consideration is stable under Scheme~1 and Scheme~2 if~\eqref{eq:maximal_stability} holds. 
\end{proof}

\begin{remark}
\label{rmk:coupling_proof}
{\em An alternative proof of stability via a coupling argument is as follows: Consider a modified scheme in which, upon arrival of each job, one server is chosen from each type uniformly at random
(i.e., $d_j=1$ for all $j \in \cal{J}$).
The job is then routed to the sampled server of type $j$
with probability $\frac{\gamma_j C_j}{\sum_{i \in \cal{J}}\gamma_i C_i}$  for each $j \in \cal{J}$.
A coupling argument, similar to the one discussed in the proof
of Theorem 3 of~\cite{Martin_AAP_1999}, shows that the system operating under the modified scheme
always has higher number of unfinished jobs than that operating under Scheme~1 or Scheme~2. 
It is easy to check that the system operating under the modified scheme is stable under~\eqref{eq:maximal_stability}.
Hence, the system operating under Scheme~1 and Scheme~2 also must be stable under~\eqref{eq:maximal_stability}.
}
\end{remark}

As discussed in~\cite{Bramson_stability_AAP}, for $\lambda > \mu \sum_{j \in \cal{J}} \gamma_j C_j$,
the system under consideration is unstable under any job assignment policy.
Thus, from Theorem~\ref{thm:stability} we conclude that Scheme~1 and Scheme~2 achieve
the maximal stability region.

\section{Mean field analysis}
\label{sec:mean_field}

We now analyze the time evolution of the number
of jobs in the system under Scheme~1 and Scheme~2.
Its exact characterization is difficult
since under both the schemes, arrivals at a given server
depend on the states of other servers. However,
it is possible to analyze the system
in the limit as the system size $N \to \infty$. Such a limit
is known as the mean field limit~\cite{Mitzenmacher_thesis,
Vvedenskaya_inftran_1996, Martin_AAP_1999} and 
it exists
because 
under random sampling of a fixed number of 
servers from each type the statistical properties of
the system do not change when 
states among servers of the same type are permuted.

To formally state our results, we define the process

\begin{equation}
\mathbf{x}_N(t)=\cbrac{x^{(j)}_{N, n}(t), j \in \mathcal{J}, n \in \mathbb{Z}_{+}} \text{ for } t \geq 0,
\end{equation}
where $x^{(j)}_{N, n}(t)$ denotes
the fraction of type $j$ servers having at least
$n$ unfinished jobs at time $t$. Thus, $\cbrac{x^{(j)}_{N, n}(t), n \in \mathbb{Z}_{+}}$ 
denotes the empirical tail distribution
of occupancy of type $j$ servers at time $t$. 
Clearly, $\mf{x}_N(t)$ is a Markov process
in the state space $\prod_{j \in \cal{J}}\UN$. 

\subsection{Convergence to the mean field}
The main aim of this subsection is to prove the following result.

\begin{theorem}
\label{thm:meanfield_conv}
If $\mf{x}_N(0)$ converges in distribution
to some constant $\mf{g} \in \Ub^M$ as $N \to \infty$, then 
the process $\cbrac{\mf{x}_N(t)}_{t \geq 0}$
converges in distribution to a process $\left\{\mf{u}(t)  \right\}_{t \geq 0}$,
lying in the space $\Ub^M$ as $N \to \infty$.
For Scheme~1, the process $\mf{u}(t)$ is given by the solution of the following system
of differential equations

\begin{align}
 \mf{u}(0) &= \mf{g}, \label{eq:diff1_x}\\
 \dot{\mf{u}}(t) &= \mf{l}(\mf{u}(t)), \label{eq:diff2_x}
\end{align}
where the mapping $\mf{l}:\Ub^M \to \brac{\mb{R}^{\mb{Z}_+}}^M$
is given by

\begin{align}
\mv{l}{0}{j}(\mf{u}) &= 0,  \text{ for } j \in \mathcal{J}, \label{eq:diff3_x}\\
 \mv{l}{k}{j}(\mf{u}) &= \frac{\lambda}{\gamma_{j}} \left(\left(\mv{u}{k-1}{j}\right)^{d_{j}} - \left(\mv{u}{k}{j}\right)^{d_{j}}\right)
\prod_{i=1}^{j-1} \left(\mv{u}{k-1}{i}\right)^{d_{i}}
\prod_{i=j+1}^{M} \left(\mv{u}{k}{i}\right)^{d_{i}} \label{eq:diff4_x}\\
&\hspace{3cm}- \mu C_{j} \left(\mv{u}{k}{j} - \mv{u}{k+1}{j}\right), \text{ for } k \geq 1, j \in \mathcal{J}. \nonumber
\end{align}
For Scheme~2, the process $\mf{u}(t)$ is given by the solution of

\begin{align}
 \mf{u}(0) &= \mf{g}, \label{eq:diff1_xc}\\
 \dot{\mf{u}}(t) &= \mf{\tilde{l}}(\mf{u}(t)), \label{eq:diff2_xc}
\end{align}
where the mapping $\mf{\tilde{l}}:\Ub^M \to \brac{\mb{R}^{\mb{Z}_+}}^M$
is given by

\begin{align}
\mv{\tilde{l}}{0}{j}(\mf{u}) &= 0, \text{ for } j \in \mathcal{J}, \label{eq:diff3_xc}\\
 \mv{\tilde{l}}{k}{j}(\mf{u}) &= \frac{\lambda}{\gamma_{j}} \left(\left(\mv{u}{k-1}{j}\right)^{d_{j}} - \left(\mv{u}{k}{j}\right)^{d_{j}}\right)
\prod_{i=1}^{j-1} \left(\mv{u}{\ktil{k-1}{j}{i}}{i}\right)^{d_{i}} \label{eq:diff4_xc}\\
&\times \prod_{i=j+1}^{M} \left(\mv{u}{\ksub{k-1}{j}{i}}{i}\right)^{d_{i}}
- \mu C_{j} \left(\mv{u}{k}{j} - \mv{u}{k+1}{j}\right), \text{ for } k \geq 1, j \in \mathcal{J}. \nonumber
\end{align}
\end{theorem}

The process $\cbrac{\mf{u}(t)}_{t\geq 0}$, defined in the
theorem above, is referred to as the {\em mean field}.
We first note that Theorem~\ref{thm:meanfield_conv} implicitly assumes
that the ordinary differential 
systems \eqref{eq:diff1_x}-\eqref{eq:diff2_x}
and~\eqref{eq:diff1_xc}-\eqref{eq:diff2_xc}
have unique solutions in the space $\Ub^M$.
In the following proposition, we show that this is indeed the case. 
To emphasize the dependence
of the solution $\mf{u}(t)$ on the initial point $\mf{g}$, we will often denote
$\mf{u}(t)$ by $\mf{u}(t,\mf{g})$.

\begin{proposition}
\label{thm:uniqueness_sol}
If $\mf{g} \in \Ub^M$, then each of the systems~\eqref{eq:diff1_x}-\eqref{eq:diff2_x}
and~\eqref{eq:diff1_xc}-\eqref{eq:diff2_xc} has a unique solution $\mf{u}(t,\mf{g}) \in \Ub^M$,
for all $t \geq 0$.
\end{proposition}  

\begin{proof}
The proof is given in Appendix~\ref{proof:uniqueness_sol}.
\end{proof}

We will prove Theorem~\ref{thm:meanfield_conv} using the theory
of semigroup operators of Markov processes as in~\cite{Vvedenskaya_inftran_1996,Martin_AAP_1999}.
Before doing so, we recall the following from~\cite{Ethier_Kurtz_book}.
\begin{itemize} 
\item For the process $\cbrac{\mf{x}_N(t)}_{t\geq 0}$, the operator semigroup
$\cbrac{\mf{T}_N(t)}_{t\geq 0}$ acting on continuous functions $f:\prod_{j=1}^{M}\UN \to \mb{R}$
is defined as

\begin{equation*}
\mf{T}_N(t)f(\mf{x})\!=\!\expect{f(\mf{x}_N(t))\vert \mf{x}_N(0)\!=\!\mf{x}} \quad \forall t \geq 0, \mf{x} \in \prod_{j \in \cal{J}} \UN.
\end{equation*}
\item For the deterministic process $\cbrac{\mf{u}(t)}_{t \geq 0}$,
the transition semigroup $\cbrac{\mf{T}(t)}_{t\geq 0}$ acting on continuous  
functions $f:\Ub^M \to \mb{R}$
is defined as

\begin{equation*}
\mf{T}(t)f(\mf{x})=f(\mf{u}(t,\mf{x})) \quad \forall t \geq 0,\mf{x} \in \Ub^M.
\end{equation*}
\end{itemize}

In the next proposition, we show 
that $\mf{T}_N(t)$ converges to $\mf{T}(t)$ uniformly
on bounded intervals. 
This
in conjunction with
Theorem 2.11 of Chapter 4 of~\cite{Ethier_Kurtz_book} proves 
Theorem~\ref{thm:meanfield_conv}.



\begin{proposition}
\label{thm:convergence_semigroup}
For both Scheme~1 and Scheme~2, and for any continuous function $f:\Ub^M \to \mb{R}$ and $t \geq 0$,

\begin{equation}
\lim_{N \rightarrow \infty} \sup_{\mf{g} \in \prod_{j \in \cal{J}}\UN} \abs{\mf{T}_N(t) f(\mf{g}) - f(\mf{u}(t,\mf{g}))}=0
\label{eq:semigroup_conv}
\end{equation}
and the convergence is uniform in $t$ within any bounded interval.
\end{proposition}

\begin{proof}
The proof is given in Appendix~\ref{proof:convergence_semigroup}.
\end{proof}

%

\begin{remark}
\label{rmk:weaker:implications}
{\em We note that Theorem~\ref{thm:meanfield_conv} implies
that if $\mf{x}_N(0) \Rightarrow \mf{g} \in \Ub^M$ as $N \to \infty$, then 
the following weaker convergence results also hold:

\begin{enumerate}
\item For each $t \geq 0$, $\mf{x}_N(t) \Rightarrow \mf{u}(t, \mf{\mf{g}})$ as $N \to \infty$.
\item For each $t \geq 0$, $j \in \cal{J}$, and $k \in \mb{Z}_+$,
$\mv{x}{N,k}{j}(t) \Rightarrow \mv{u}{k}{j}(t,\mf{g})$ as $N \to \infty$.
\item For each $t \geq 0$, $j \in \cal{J}$, and $k \in \mb{Z}_+$,
$\expect{\mv{x}{N,k}{j}(t)} \to \mv{u}{k}{j}(t,\mf{g})$ as $N \to \infty$.
\end{enumerate}
The last assertion follows from the first since 
$\mv{x}{N,k}{j}(t)$ is bounded for each $N, j,k,t$.}
\end{remark}

\subsection{Properties of the mean field}
\label{sec:mean_field_properties}

In this section, we characterize some important properties of the mean field.
In particular, we show that, under the stability condition~\eqref{eq:maximal_stability},
both~\eqref{eq:diff1_x}-\eqref{eq:diff2_x} and~\eqref{eq:diff1_xc}-\eqref{eq:diff2_xc}
have unique equilibrium points in $\U^M$. Further, we show that
the equilibrium points are globally asymptotically stable for both systems.

Let  $\mf{P}$, $\tilde{\mf{P}}$ denote the equilibrium points
of~\eqref{eq:diff1_x}-\eqref{eq:diff2_x} and~\eqref{eq:diff1_xc}-\eqref{eq:diff2_xc}, respectively.
In other words, $\mf{P}$  and $\tilde{\mf{P}}$ satisfy
$\mf{l}(\mf{P})=\mf{0}$ and $\mf{\tilde{l}}({\mf{\tilde{P}}})=\mf{0}$.
Hence, for all $k \in \mb{Z}_+$
and $j \in \cal{J}$ the following must hold

\begin{multline}
P_{k+1}^{(j)}-\mv{P}{k+2}{j} = \Delta_j\left(\left(\mv{P}{k}{j}\right)^{d_{j}} - \left(\mv{P}{k+1}{j}\right)^{d_{j}}\right)\\
\times \prod_{i=1}^{j-1} \left(\mv{P}{k}{i}\right)^{d_{i}}
  \prod_{i=j+1}^{M} \left(\mv{P}{k+1}{i}\right)^{d_{i}},
\label{eq:tailhet_x}
\end{multline}

\begin{multline}
\tilde{P}_{k+1}^{(j)}-\mv{\tilde{P}}{k+2}{j} = \Delta_j\left(\left(\mv{\tilde{P}}{k}{j}\right)^{d_{j}} - \left(\mv{\tilde{P}}{k+1}{j}\right)^{d_{j}}\right)\\
\times \prod_{i=1}^{j-1} \left(\mv{\tilde{P}}{\ktil{k}{j}{i}}{i}\right)^{d_{i}}
\prod_{i=j+1}^{M} \left(\mv{\tilde{P}}{\ksub{k}{j}{i}}{i}\right)^{d_{i}},
\label{eq:tailhet_xc}
\end{multline} 
where $\Delta_j=\frac{\lambda}{\mu \gamma_j C_j}$ for each $j \in \cal{J}$.
Note that by definition we have $\mv{P}{0}{j}=\mv{\tilde{P}}{0}{j}=1$ 
for all $j \in \cal{J}$.
The next proposition reveals an important property of the equilibrium
points $\mf{P}$ and $\tilde{\mf{P}}$. To state it we first need the following definition.

\begin{definition}
A real sequence $\cbrac{z_n}_{n\geq 1}$ is said to decrease doubly
exponentially if and only if there exist positive
constants $L$, $\omega < 1$, $\theta > 1$, and $\kappa$ such that
$z_n \leq \kappa \omega^{\theta^{n}}$ for all $n \geq L$.
\end{definition}

Hence,  if a sequence $\cbrac{z_n}_{n\geq 1}$ decays doubly exponentially,
then it is summable, i.e., $\sum_{n=1}^{\infty} z_n < \infty$.

\begin{proposition}
\label{thm:tail_het_properties}
Assume that for each $j \in \cal{J}$, 
$\mv{P}{k}{j}, \mv{\tilde{P}}{k}{j} \downarrow 0$ as $k \rightarrow \infty$. 
Then the following equations must hold

\begin{equation}
\sum_{j \in \cal{J}} \frac{ \mv{P}{l+1}{j}}{\Delta_j}= \prod_{j\in \cal{J}}\brac{\mv{P}{l}{j}}^{d_j}.
\label{eq:coupling_x}
\end{equation}
\begin{equation}
\frac{ \mv{\tilde{P}}{l+1}{1}}{\Delta_1} + \sum_{j=2}^{M} \frac{\mv{\tilde{P}}{\ksub{l-1}{1}{j}+1}{j}}{\Delta_j} = \left(\mv{\tilde{P}}{l}{1}\right)^{d_{1}} \prod_{j=2}^{M} \left(\mv{\tilde{P}}{\ksub{l-1}{1}{j}}{j}\right)^{d_{j}}.
\label{eq:coupling_xc}
\end{equation}
Further, for each $j \in \cal{J}$, the sequences $\cbrac{\mv{P}{k}{j}, k \in \mb{Z}_+}$ 
and $\cbrac{\mv{\tilde{P}}{k}{j}, k \in \mb{Z}_+}$ decrease doubly exponentially.
In particular, under the assumption of the proposition,
both $\cbrac{\mv{P}{k}{j}, k \in \mb{Z}_+}$ 
and $\cbrac{\mv{\tilde{P}}{k}{j}, k \in \mb{Z}_+}$ are summable sequences.
\end{proposition}

\begin{proof} 
We prove the proposition for $\mf{P}$. The proof for $\tilde{\mf{P}}$
follows along the same line of arguments. 
For a fix $j$ we add~\eqref{eq:tailhet_x} for all $k \geq l$ and use 
$\lim_{k \to \infty} \mv{P}{k}{j}=0$ to obtain

\begin{equation}
\mv{P}{l+1}{j}= \Delta_j \sum_{k\geq l} \sbrac{\prod_{i=1}^{j}\brac{\mv{P}{k}{i}}^{d_i}\prod_{i=j+1}^{M}\brac{\mv{P}{k+1}{i}}^{d_i}-\prod_{i=1}^{j-1}\brac{\mv{P}{k}{i}}^{d_i}\prod_{i=j}^{M}\brac{\mv{P}{k+1}{i}}^{d_i}}
\end{equation}
Now, multiplying both sides of the above equation by $\frac{1}{\Delta_j}$
and adding over all $j \in \cal{J}$ and using $\lim_{k \to \infty} \mv{P}{k}{j}=0$
yields~\eqref{eq:coupling_x}.
From~\eqref{eq:coupling_x} we obtain
$\frac{\mv{P}{k+1}{j}}{\Delta_j} \leq \prod_{j\in \cal{J}} \brac{\mv{P}{k}{j}}^{d_j}
\leq \brac{\hat{P}_k}^d$,
%
where $\hat{P}_k=\max_{1 \leq j \leq M} \mv{P}{k}{j}$ and $d=\sum_{j\in \cal{J}} d_j$. 
Thus, we have $\mv{P}{k+1}{j} \leq \delta \hat{P}_k$, 
where $\delta=\brac{\hat{P}_k}^{d-1} \max_{1 \leq j \leq M} (\Delta_j)$. 
Since by hypothesis, for each $j$, $\mv{P}{k}{j} \rightarrow 0$ 
as $k \rightarrow \infty$, one can choose $k$ sufficiently large such that $\delta < 1$. Hence, we have $\brac{\max_{1 \leq j \leq M} \mv{P}{k+1}{j}} \leq \delta \hat{P}_k$.
Similarly we have, $\brac{\max_{1 \leq j \leq M} \mv{P}{k+n}{j}} \leq \delta^{\frac{{d^n-1}}{d-1}} \hat{P}_k$. This proves that the sequence $\cbrac{\mv{P}{k}{j}, k \in \mb{Z}_+}$
decreases doubly exponentially for each $j$.
\end{proof}

The following proposition
guarantees that there exists  equilibrium points of  
systems~\eqn{eq:diff1_x}-\eqn{eq:diff2_x}
and~\eqn{eq:diff1_xc}-\eqn{eq:diff2_xc} in $\U^M$ .

\begin{theorem}
\label{thm:existence_fixedp}
Under condition~\eqref{eq:maximal_stability}, there exists an 
equilibrium
point $\mf{P}$ of the system~\eqref{eq:diff1_x}-\eqref{eq:diff2_x} 
and $\tilde{\mf{P}}$ of the system~\eqref{eq:diff1_xc}-\eqref{eq:diff2_xc}  
in the space $\U^M$.
\end{theorem}

\begin{proof}
The proof is given in Appendix~\ref{proof:existence_fixedp}.
\end{proof}

The next theorem shows that $\mf{P}$ and $\tilde{\mf{P}}$
are the unique globally asymptotically stable equilibrium points
of the systems~\eqref{eq:diff1_x}-\eqref{eq:diff2_x} and~\eqref{eq:diff1_xc}-\eqref{eq:diff2_xc}
in the space $\U^M$.

\begin{theorem}
\label{thm:uniqueness_fixedp}
Under condition~\eqref{eq:maximal_stability},

\begin{equation}
\lim_{t \rightarrow \infty} \mf{u}(t, \mf{g})=\mf{P}\in \U^M \text{ for all } \mf{g} \in \U^M,
\label{eq:exp_conv1}
\end{equation}
for Scheme~1 and

\begin{equation}
\lim_{t \rightarrow \infty} \mf{u}(t, \mf{g})=\tilde{\mf{P}}\in \U^M \text{ for all } \mf{g} \in \U^M,
\label{eq:exp_conv2}
\end{equation}
for Scheme~2. Hence, $\mf{P}$ and $\tilde{\mf{P}}$ are globally asymptotically stable fixed points
of systems~\eqref{eq:diff1_x}-\eqref{eq:diff2_x} and~\eqref{eq:diff1_xc}-\eqref{eq:diff2_xc},
respectively. Furthermore, $\mf{P}$ and $\tilde{\mf{P}}$ are the only equilibrium points
of the above systems in the space $\U^M$.
\end{theorem}

\begin{proof}
The proof for Scheme~1 is given in Appendix~\ref{proof:uniqueness_fixedp}.
For Scheme~2, the theorem can be similarly proved.
\end{proof}

We now show that, under~\eqref{eq:maximal_stability},
the stationary distribution
of the process $\mf{x}_N$ converges
weakly to the Dirac measure concentrated at
the unique equilibrium point of the mean field.
Let $\pi_N$ denote the stationary distribution
of the process $\mf{x}_N$. Clearly, $\pi_N$
exists and is unique under~\eqref{eq:maximal_stability}.
Further, for each fixed $N$, $\mf{x}_N(t) \Rightarrow \mf{x}_N(\infty)$
as $t \to \infty$, where $\mf{x}_N(\infty)$ is a random variable
distributed as $\pi_N$.

\begin{theorem}
\label{thm:stationary}
Under condition~\eqref{eq:maximal_stability},
we have

\begin{equation}
\pi_N \Rightarrow \delta_{\mf{P}},
\end{equation}
for Scheme~1 and

\begin{equation}
\pi_N \Rightarrow \delta_{\tilde{\mf{P}}},
\end{equation} 
for Scheme~2. 
\end{theorem}

\begin{proof}
We prove the theorem for Scheme~1. The proof for Scheme~2
follows similarly.

Note that since the space $\Ub^M$ is compact, so is the space
of probability measures on $\Ub^M$. Therefore, the sequence
of probability measures $\cbrac{\pi_N}_N$ has limit points.
Thus, in order to prove the theorem we need to show that all 
limit points coincide with $\delta_{\mf{P}}$.

Due to Theorem~\ref{thm:meanfield_conv}, any limit point
$\pi$ of the sequence $\pi_N$ must be an invariant distribution
of the maps $\mf{g} \mapsto \mf{u}(t,\mf{g})$. Hence,
by uniqueness proved in Theorem~\ref{thm:uniqueness_fixedp},
it is sufficient to prove that $\pi$ is concentrated on $\U^M$.
To prove that $\pi$ is concentrated on $\U^M$
it is sufficient to show that $\mb{E}_{\pi}\sbrac{\sum_{n\geq 1} \mv{g}{n}{j}} < \infty$
for all $j \in \cal{J}$. The coupling described in Remark~\ref{rmk:coupling_proof} implies
that $\mb{E}_{\pi_N}\sbrac{{\sum_{n\geq 1} \mv{g}{n}{j}}} \leq \frac{\rho}{1-\rho}$,
where $\rho=\frac{\lambda}{\mu \sum_{j \in \cal{J}} \gamma_j C_j} < 1$.
Hence, $\mb{E}_{\pi_N}\sbrac{{\sum_{n\geq 1} \mv{g}{n}{j}}} \to \mb{E}_{\pi}\sbrac{\sum_{n\geq 1} \mv{g}{n}{j}} \leq \frac{\rho}{1-\rho}$. This completes the proof.
\end{proof}

We have so far established that the interchange property
indicated in Figure~2 holds.
\begin{figure}
\begin{center}
\begin{tikzpicture}
\label{commdiagram}
\centering
  \matrix (m) [matrix of math nodes,row sep=3cm,column sep=3cm,minimum width=2em] {
     \mf{x}_N(t) & \mf{u}(t) \\
     \mf{x}_N(\infty) & \mf{P} \\};
  \path[-stealth]
    (m-1-1) edge node [left] {$t \to \infty$}node[below,rotate=90] {Theorem~\ref{thm:stability}}(m-2-1)
            edge  node [above] {$N \to \infty$}node [below]{Theorem~\ref{thm:meanfield_conv}} (m-1-2)
    (m-2-1.east|-m-2-2) edge node [below] {$N \to \infty$}node[above]{Theorem~\ref{thm:stationary}}  (m-2-2)
    (m-1-2) edge node [right] {$t \to \infty$} node [above,rotate=90]{Theorem~\ref{thm:uniqueness_fixedp}}(m-2-2);
\end{tikzpicture}
\caption{Commutativity of limits}
\end{center}
\end{figure}
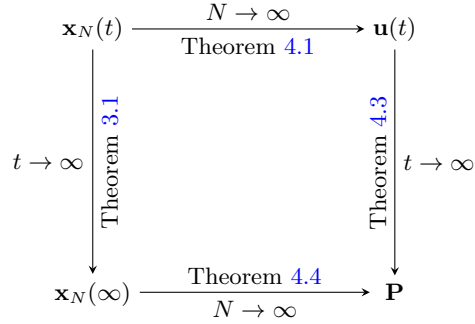
Note that the convergences indicated in the figure are in distribution.

\subsection{Propagation of chaos}

In this subsection, we focus on the 
occupancies of a given finite set of servers
as $N \to \infty$. We show that as the system size grows
the server occupancies become independent of each other.
Such independence holds at any finite time and also
at the equilibrium, provided that the initial server occupancies
satisfy certain assumptions.
This is formally known as the {\em propagation of chaos}~\cite{Graham,Sznitman}
or {\em asymptotic independence property}~\cite{Bramson_asymp_indep, Bramson_randomized_load_balancing} 
in the literature.

To formally state the results we introduce the following notations.
Let $\mv{q}{N}{j,k}(t)$, for $j \in\cal{J}$ and $k \in \cbrac{1,2,\ldots,N\gamma_j}$, 
denote the occupancy of the $k^{\textrm{th}}$ server of type $j$ at time $t \geq 0$.
By  $\mv{q}{N}{j,k}(\infty)$ we denote the occupancy of the 
$k^{\textrm{th}}$ server of type $j$ in equilibrium. Further,
let $\mv{\chi}{N,n}{j}(t)$, for $j \in\cal{J}$ and $n \in \mb{Z}_+$,
denote the fraction of type $j$ servers having occupancy $n$ at time
$t \geq 0$. Define the process $\bs{\chi}_N(t)=\cbrac{\mv{\chi}{N,n}{j}(t), j \in\mathcal{J}, n\in \mb{Z}_+}$.
Clearly, $\chi_N^{(j)}(t)=\cbrac{\mv{\chi}{N,n}{j}(t), n\in \mb{Z}_+}$
denotes the empirical distribution of occupancies of type $j$ servers and 
for each $n,j$, we have $\mv{\chi}{N,n}{j}(t)=\mv{x}{N,n}{j}(t)-\mv{x}{N,(n+1)}{j}(t)$.
By ${\chi}_N^{(j)}(\infty)$ we will denote the empirical distribution 
occupancies for type $j$ servers in equilibrium. Let the process 
$\mf{Q}(t)=\cbrac{\mv{Q}{n}{j}(t), j \in\mathcal{J}, n\in \mb{Z}_+}$
be defined as $\mv{Q}{n}{j}(t)=\mv{u}{n}{j}(t)-\mv{u}{n+1}{j}(t)$, for $t \in [0,\infty]$.
Further, we denote by $Q^{(j)}(t)$ the distribution on $\mb{Z}_+$
given by $Q^{(j)}(t)=\cbrac{\mv{Q}{n}{j}, n \in \mb{Z}_+}$.
We also define the following notion
of exchangeable random variables.

\begin{definition}
Let $\cbrac{\mv{q}{N}{j,k}, 1 \leq k \leq N \gamma_j, 1 \leq j \leq M}$ 
denote a collection of $N$ random variables among which $N \gamma_j$
belong to a particular class $j$ and are indexed by $k$, where $1 \leq k \leq N\gamma_j$.
The collection is called intra-class exchangeable if the joint law
of the collection is invariant under permutation of indices, $1 \leq k \leq N\gamma_j$,
of random variables belonging to the same class. 
\end{definition}

\begin{proposition}
\label{thm:chaos}
For the model considered in this paper, for both schemes,  $\cbrac{\mv{q}{N}{j,k}(0), 1 \leq k \leq N \gamma_j, 1 \leq j \leq M}$
is intra-class exchangeable and if $\mf{x}_N(0) \Rightarrow \mf{g}\in \U^M$ as
$N \to \infty$, then the following holds

\begin{enumerate}[(i)]
\item  For each fix $k$ and $t \in [0, \infty]$,
$\mv{q}{N}{j,k}(t) \Rightarrow U^{(j)}(t)$ as $N \to \infty$,
where $U^{(j)}(t)$ is a random variable with distribution
$Q^{(j)}(t)$.

\item  Fix positive integers $r_1,r_2, \ldots, r_M$.
For each $t \in [0, \infty]$,

\begin{equation*}
\cbrac{\mv{q}{N}{j,k}, 1 \leq k \leq r_j, 1 \leq j \leq M} \Rightarrow \cbrac{U^{(j,k)}(t),1 \leq k \leq r_j, 1 \leq j \leq M},
\end{equation*}
as $N \to \infty$, where $U^{(j,k)}(t)$, $1 \leq k \leq r_j, 1 \leq j \leq M$, are independent random variables
with $U^{(j,k)}(t)$ having distribution  $Q^{(j)}(t)$ for all $1 \leq k \leq r_j$. 
\end{enumerate}
\end{proposition}

\begin{proof}
Note that the first part of the proposition is a special case of the second part.
Hence, it is sufficient to prove the second part.
We will provide a proof  for the $M=2$ case. 
The proof can be readily generalized to any $M \geq 2$. 

Due to the dynamics of the system (under Scheme~1 or Scheme~2)
and the hypothesis of the proposition 
$\{\mv{q}{N}{j,k}(t), 1 \leq k \leq N \gamma_j, 1 \leq j \leq M\}$
is intra-class exchangeable for all $t \in [0,\infty]$.
The hypothesis of the proposition also 
implies that $\bs{\chi}_N(t) \Rightarrow \mf{Q}(t)$ as $N \to \infty$
for all $t \in [0,\infty]$.
Henceforth, we will omit the variable $t$ in our calculations, which
hold for all $t \in [0,\infty]$.

To prove the proposition, it is sufficient to show
that the following convergence holds as $N \to \infty$.

\begin{equation}
\expect{\prod_{k=1}^{r_1} \phi_k\brac{\mv{q}{N}{1,k}}\prod_{k=1}^{r_2} \psi_k\brac{\mv{q}{N}{2,k}}}
\to \prod_{k=1}^{r_1} \inprod{\phi_k}{Q^{(1)}} \prod_{k=1}^{r_2} \inprod{\psi_k}{Q^{(2)}}
\end{equation}
for all bounded mappings $\phi_k,\psi_k:\mb{Z}_+ \to \mb{R}_+$.
Now we have

\begin{multline}
\abs{\expect{\prod_{k=1}^{r_1} \phi_k\brac{\mv{q}{N}{1,k}}\prod_{k=1}^{r_2} \psi_k\brac{\mv{q}{N}{2,k}}}
- \prod_{k=1}^{r_1} \inprod{\phi_k}{Q^{(1)}} \prod_{k=1}^{r_2} \inprod{\psi_k}{Q^{(2)}}}\\
\leq \abs{\expect{\prod_{k=1}^{r_1} \phi_k\brac{\mv{q}{N}{1,k}}\prod_{k=1}^{r_2} \psi_k\brac{\mv{q}{N}{2,k}}}
-\expect{\prod_{k=1}^{r_1} \inprod{\phi_k}{\mv{\chi}{N}{1}} \prod_{k=1}^{r_2} \inprod{\psi_k}{\mv{\chi}{N}{2}}}}\\
+\abs{\expect{\prod_{k=1}^{r_1} \inprod{\phi_k}{\mv{\chi}{N}{1}} \prod_{k=1}^{r_2} \inprod{\psi_k}{\mv{\chi}{N}{2}}}-\prod_{k=1}^{r_1} \inprod{\phi_k}{Q^{(1)}} \prod_{k=1}^{r_2} \inprod{\psi_k}{Q^{(2)}}}.
\label{eq:neq}
\end{multline}
Note that the second term on the right hand side of the above inequality vanishes
as $N \to \infty$ since $\mv{\chi}{N}{j} \Rightarrow Q^{(j)}$  as $N \to \infty$ for
$j=1,2$ and $Q^{(1)}$ and $Q^{(2)}$ are constants.  Now, due to exchangeability
we have

\begin{multline}
\expect{\prod_{k=1}^{r_1} \phi_k\brac{\mv{q}{N}{1,k}}\prod_{k=1}^{r_2} \psi_k\brac{\mv{q}{N}{2,k}}}
=\frac{1}{(N \gamma_1)_{r_1}(N \gamma_2)_{r_2}}\\
\times \expect{\sum_{\sigma \in P(r_1,N\gamma_1)}\sum_{\sigma' \in P(r_1,N\gamma_1)}
\prod_{k=1}^{r_1} \phi_k\brac{\mv{q}{N}{1,\sigma(k)}}\prod_{k=1}^{r_2} \psi_k\brac{\mv{q}{N}{2,\sigma'(k)}}},
\end{multline}
where $(N)_k=N(N-1)\ldots(N-k+1)$, and $P(r,n)$ denotes the set of all permutations
of the numbers $\cbrac{1,2,\ldots,N}$ taken $r$ at a time. Also,
by definition of $\mv{\chi}{N}{j}$ we have

\begin{multline}
\expect{\prod_{k=1}^{r_1} \inprod{\phi_k}{\mv{\chi}{N}{1}} \prod_{k=1}^{r_2} \inprod{\psi_k}{\mv{\chi}{N}{2}}}
=\mb{E}\left[\brac{\prod_{k=1}^{r_1} \frac{1}{N \gamma_1}\sum_{l=1}^{N \gamma_1} \phi_k\brac{\mv{q}{N}{1,l}}}\right.\\
\left. \times \brac{\prod_{k=1}^{r_2} \frac{1}{N \gamma_2}\sum_{l=1}^{N \gamma_2} \psi_k\brac{\mv{q}{N}{2,l}}}\right]
\end{multline}
Hence, the first term on the right hand side of~\eqref{eq:neq}
can be bounded as follows

\begin{align*}
\left \vert \expect{\prod_{k=1}^{r_1} \phi_k\brac{\mv{q}{N}{1,k}}\prod_{k=1}^{r_2} \psi_k\brac{\mv{q}{N}{2,k}}}\right.
&-\left. \expect{\prod_{k=1}^{r_1} \inprod{\phi_k}{\mv{\chi}{N}{1}} \prod_{k=1}^{r_2} \inprod{\psi_k}{\mv{\chi}{N}{2}}}\right \vert\\
&\leq 2B^{r_1+r_2} \brac{1- \frac{(N\gamma_1)_{r_1}(N\gamma_2)_{r_2}}{(N \gamma_1)^{r_1}(N \gamma_2)^{r_2}}},\\
&\to 0 \text{ as } N \to \infty,
\end{align*}
where $\max\brac{\norm{\phi_k}_{\infty},\norm{\psi_k}_{\infty}} = B$.
This completes the proof.
\end{proof}

Thus, the above proposition shows that in the limiting system
server occupancies become independent of each other. It also
shows that the stationary occupancy distribution
of any type $j$ server is given by 
$Q^{(j)}(\infty)=\cbrac{\mv{P}{n}{j}-\mv{P}{n+1}{j}, n \in \mb{Z}_+}$ for Scheme~1
and $Q^{(j)}(\infty)=\cbrac{\mv{P}{n}{j}-\mv{\tilde{P}}{n+1}{j}, n \in \mb{Z}_+}$ for Scheme~2.

\section{Computation of the stationary distribution}
\label{sec:computation}

So far we have shown that in the limiting system ($N \to \infty$) 
each finite collection of servers behave independently and 
the stationary tail distribution of occupancy of a type $j \in \cal{J}$
server in the limiting system is given by $\cbrac{\mv{P}{k}{j}, k \in \mb{Z}_+}$ under Scheme~1
and $\cbrac{\mv{P}{k}{j}, k \in \mb{Z}_+}$ under Scheme~2.
Using the independence of servers in the limiting system
we conclude the following proposition.

\begin{proposition}
\label{thm:lambdak}
In equilibrium, the arrival process of jobs
at any given server in the limiting system
is a state dependent Poisson process. Further,
the arrival rate of jobs to a server
of type $j \in \cal{J}$ when it has occupancy
$k$ in the equilibrium is given by

\begin{equation}
\label{eq:lambdak_x}
\mv{\lambda}{k}{j}= \frac{\lambda}{\gamma_j}\frac{\brac{\mv{P}{k}{j}}^{d_j}-\brac{\mv{P}{k+1}{j}}^{d_j}}{\mv{P}{k}{j}-\mv{P}{k+1}{j}} \prod_{i=1}^{j-1} \brac{\mv{P}{k}{i}}^{d_i}\prod_{i=j+1}^{M} \brac{\mv{P}{k+1}{i}}^{d_i},
\end{equation}
for Scheme~1 and

\begin{equation}
\label{eq:lambdak_xc}
\mv{\tilde{\lambda}}{k}{j}= \frac{\lambda}{\gamma_j}\frac{\brac{\mv{\tilde{P}}{k}{j}}^{d_j}-\brac{\mv{\tilde{P}}{k+1}{j}}^{d_j}}{\mv{\tilde{P}}{k}{j}-\mv{\tilde{P}}{k+1}{j}} \prod_{i=1}^{j-1} \brac{\mv{\tilde{P}}{\ktil{k}{j}{i}}{i}}^{d_i}\prod_{i=j+1}^{M} \brac{\mv{\tilde{P}}{\ksub{k}{j}{i}}{i}}^{d_i},
\end{equation}
for Scheme~2.
\end{proposition}

\begin{proof}
We provide the proof for Scheme~1. The proof for Scheme~2 follows from similar
line of arguments.

Consider a {\em tagged} type $j$ server in the system
and the arrivals that have the tagged server 
as one of its possible 
destinations. These arrivals constitute the {\em potential arrival 
process} at the tagged server. The probability 
that the tagged server is selected as a potential destination server for
a new arrival is $\frac{\binom{N\gamma_j-1}{d_j-1}}{\binom{N \gamma_j}{d_j}}=\frac{d_j}{N \gamma_j}$. 
Thus, {due to Poisson thinning},
the potential arrival process to the tagged server is a Poisson process with rate 
$\frac{d_j}{N \gamma_j}\times N\lambda=\frac{d_j \lambda}{\gamma_j}$.

Next, we consider the arrivals that actually join the tagged server. 
These arrivals constitute the actual arrival process 
at the server. For finite $N$, this process is 
not Poisson since a potential arrival to the tagged server actually 
joins the server depending on the number of jobs present at 
the other possible destination servers. However, as 
$N \rightarrow \infty$, due to the asymptotic independence 
property shown in~\ref{thm:chaos} the occupancies of the sampled servers 
become independent of each 
other. As a result, {in equilibrium} the actual arrival process 
converges to a state dependent Poisson process as $N \rightarrow \infty$.

Consider the potential arrivals that occur to the tagged server 
when its occupancy is $k$. 
This arrival actually joins the tagged server with 
probability $\frac{1}{x+1}$ when $x$ other servers
among the $d_j$ servers of type $j$ have occupancy $k$,
all the $d_i$ servers of type $i < j$ have at least occupancy $k$,
and   all the $d_i$ servers of type $i > j$ have at least occupancy $k+1$.
Thus, the total arrival rate $\mv{\lambda}{k}{j}$ can be computed as

\begin{multline}
\mv{\lambda}{k}{j}=\frac{d_j \lambda}{\gamma_j} \sum_{x=0}^{d_j-1} \frac{1}{x+1}\binom{d_j-1}{x}\brac{\mv{P}{k}{j}-\mv{P}{k+1}{j}}^x \brac{\mv{P}{k+1}{j}}^{d_j-1-x} \\
\times \prod_{i=1}^{j-1} \brac{\mv{P}{k}{i}}^{d_i}\prod_{i=j+1}^{M} \brac{\mv{P}{k+1}{i}}^{d_i},
\end{multline}
which simplifies to~\eqref{eq:lambdak_x}.
\end{proof}
Hence, the above proposition shows that in equilibrium the arrival rate
at a given server depends on the stationary tail probabilities $\mv{P}{k}{j}$,
$k \in \mb{Z}_+$ and $j \in \cal{J}$. 
 
The stationary tail probabilities can in turn be expressed
as functions of the arrival rate. Indeed, in 
equilibrium the global balance equations (which
hold under state dependent Poisson arrivals due to
Theorems 3.10 and 3.14 of~\cite{Kelly_book})
yield

\begin{equation}
\mv{\pi}{k}{j}\mv{\lambda}{k}{j}=\mv{\pi}{k+1}{j} \mu C_j, \text{ for } j \in \mathcal{J}, k \in \mb{Z}_+,
\label{eq:global_balance_x}
\end{equation}
where $\mv{\pi}{k}{j}=\mv{P}{k}{j}-\mv{P}{k+1}{j}$.
Hence, the equilibrium point $\mf{P}$
is the unique fixed point of the mapping $\Theta : \U^M \to \U^M$
defined as $\Theta(\mf{P})= F(G(\mf{P}))$, where $G(\cdot)$
denotes the mapping from $\U^M$ to the space of possible arrival rates
(defined by~\eqref{eq:lambdak_x}) and
$F(\cdot)$ denotes the mapping from the space of possible arrival rates 
to the space $\U^M$ (defined by~\eqref{eq:global_balance_x}).
Thus, the equilibrium point $\mf{P}$ can be computed
using the fixed point iterations (i.e., by repeatedly applying the mapping~$\Theta(\cdot)$
to some arbitrary point $\mf{Q} \in \U^M$.)

\begin{remark}
\label{rmk:insensitivity}
{\em So far our results have been obtained 
for exponential job length distributions.
Note that the conclusions of Proposition~\ref{thm:lambdak} continue to 
hold for any job length distributions due to the Whittle balance criterion \cite{Whittle_JAP_1985} that can be shown to hold for the stationary distribution (also see Theorems 3.10 and 3.14 of~\cite{Kelly_book}).
In view of the uniqueness of the stationary distribution and propagation of chaos this suggests that in stationarity the servers are asymptotically independent for general job size distributions.
In Section~\ref{sec:numerics}, we provide numerical evidence
to support insensitivity.}
\end{remark}

\begin{remark}
\label{rmk:resp_time}
{\em From Proposition~\ref{thm:chaos} it directly follows that
the expected occupancy of a type $j$ server at equilibrium
is given by $\sum_{k=1}^{\infty} P_k^{(j)}$ for Scheme~1
and $\sum_{k=1}^{\infty} \tilde{P}_k^{(j)}$ for Scheme~2.
Hence, a simple application of the Little's law, yields that the
mean sojourn time of jobs
in the limiting system is given by
$\bar{T}=\frac{1}{\lambda}\sum_{j=1}^{M} \sum_{k=1}^{\infty} \gamma_j P_k^{(j)}$
for Scheme~1 and $\bar{T}=\frac{1}{\lambda}\sum_{j=1}^{M} \sum_{k=1}^{\infty} \gamma_j \tilde{P}_k^{(j)}$
for Scheme~2. Thus, the mean sojourn time of jobs in the limiting
system can be computed using stationary tail probabilities
which in turn can be computed using the fixed point method described in this section.}
\end{remark}

\section{Numerical Results}
\label{sec:numerics}

In this section, we present simulation results
to compare the different job assignment schemes
discussed in this paper. The results also 
indicate the accuracy of the asymptotic analyses
of the Scheme~1 and Scheme~2  in
predicting their performance in a finite system
of servers. We set $\mu=1$ in all our simulations.

\begin{figure}
 \centering
 \includegraphics[height=0.5\columnwidth,
  keepaspectratio]{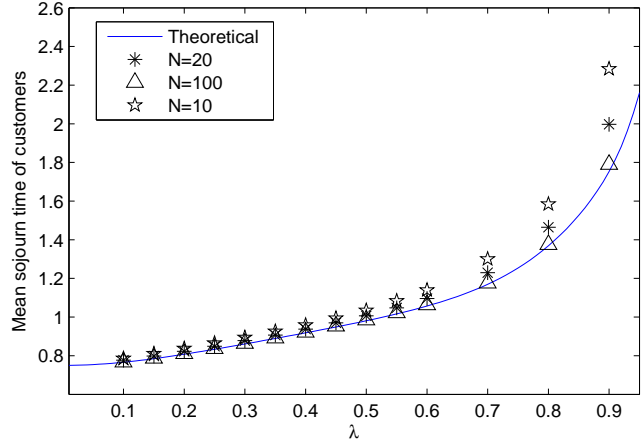}
 \caption{Mean sojourn time jobs as a function of $\lambda$ for different values of $N$.
 We set $C_1=2/3$, $C_2=4/3$,  and $\gamma_1=\gamma_2=0.5$.}
 \label{fig:accuracy}
 \end{figure}
 
To determine accuracy of the asymptotic analysis presented
in the paper we first compare the results obtained from the theoretical analysis
with that obtained from the simulations.
In Figure~\ref{fig:accuracy}, we plot the mean sojourn time jobs
as a function of the normalized arrival rate, $\lambda$, for different
values of the system size $N$. We observe a very good match between
the analysis and simulation results for $N=100$.
For $N=10$ and $N= 20$ the relative errors between
the analysis and the simulation results are around 10\% and 5\%, respectively.
Thus, we conclude that the asymptotic analysis accurately captures the behaviour of the system
for moderately large system sizes.

 \begin{figure}
 \centering
 \includegraphics[height=0.5\columnwidth,
  keepaspectratio]{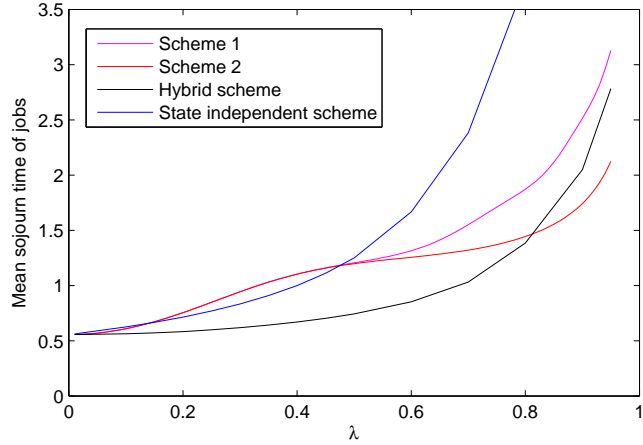}
 \caption{Mean sojourn time jobs as a function of $\lambda$ for different schemes.
 We set $M=2$, $C_1=1/5$, $C_2=9/5$,  $\gamma_1=\gamma_2=0.5$, and $d_1=d_2=2$.
 Routing probabilities for the state independent scheme and the Hybrid SQ($d$) scheme 
 are optimized based on $\lambda$.}
 \label{fig:comparison1}
 \end{figure}
 
 We now compare the performance of the proposed schemes
 with that of other existing schemes for heterogeneous scenario.
 In particular, we consider the following two schemes as benchmarks.
 
 \subsection{The state independent scheme}

As a baseline, we consider a scheme that assigns  an incoming job to a
server with a fixed probability, independent of the 
current state of the servers in the system~\cite{Altman_Opt_load_2008}. 
We denote by $p_j$, for $j \in \cal{J}$, 
the probability with which an arrival is assigned to one
of the servers of type $j$. The probabilities
$p_j$, $j \in \cal{J}$, can be chosen chosen such that
the mean sojourn time of the jobs is minimized.
Clearly, in this scheme, no communication is required between the
job dispatcher and the servers as the job assignment decisions 
are made independently of the state of the servers.

\subsection{The hybrid SQ($d$) scheme}

In this scheme~\cite{Mukhopadhyay_ITC_2014}, upon arrival of a new job,
the router first chooses a server type $j \in \cal{J}$
with probability $p_j$.
Then $d_j$ servers of type $j$
are chosen uniformly at random from set of $N\gamma_j$
servers of type $j$. The job is then
assigned to the server having the least number of unfinished
jobs among the $d_j$ chosen servers. Ties are broken by tossing
a fair coin. As in the state independent scheme, the probabilities $p_j$,
$j \in \cal{J}$, can be chosen  such  that the mean sojourn time
of jobs in the system is minimized.  

 We choose the parameter values as follows: $M=2$, $C_1=1/5$, $C_2=9/5$,  $\gamma_1=\gamma_2=0.5$,
 and $d_1=d_2=2$.
 Under this parameter setting, the stability region for all the
 schemes under consideration is $\lambda < 1$. 
 In Figure~\ref{fig:comparison1}, we plot
 the mean sojourn time of jobs as a function 
 of the normalized arrival rate, $\lambda$, for 
 Scheme~1, Scheme~2, the state independent scheme,
 and the hybrid SQ($d$) scheme. We choose the optimal routing probabilities
 $p_j$, $j \in \cal{J}$, for both state independent
 scheme and the hybrid SQ($d$) scheme.
 We observe that the mean sojourn time of jobs under Scheme~1
 and is almost the same as that under Scheme~2  for small values of $\lambda$.
 However, for larger values of $\lambda$, Scheme~2 outperforms Scheme~1.
 This is expected for reasons explained in Section~\ref{sec:model}. We also
 see that hybrid SQ($d$) scheme results in a smaller mean sojourn time of jobs
 than that in Scheme~1 and Scheme~2, for smaller values of $\lambda$.
 This is because, in the hybrid SQ($4$) scheme, the routing probabilities
 are chosen optimally based on the arrival rate $\lambda$. However, for larger
 values of $\lambda$, we observe that Scheme~2 outperforms the hybrid SQ($d$)
 scheme.
 
\begin{figure}
 \centering
 \includegraphics[height=0.5\columnwidth,
  keepaspectratio]{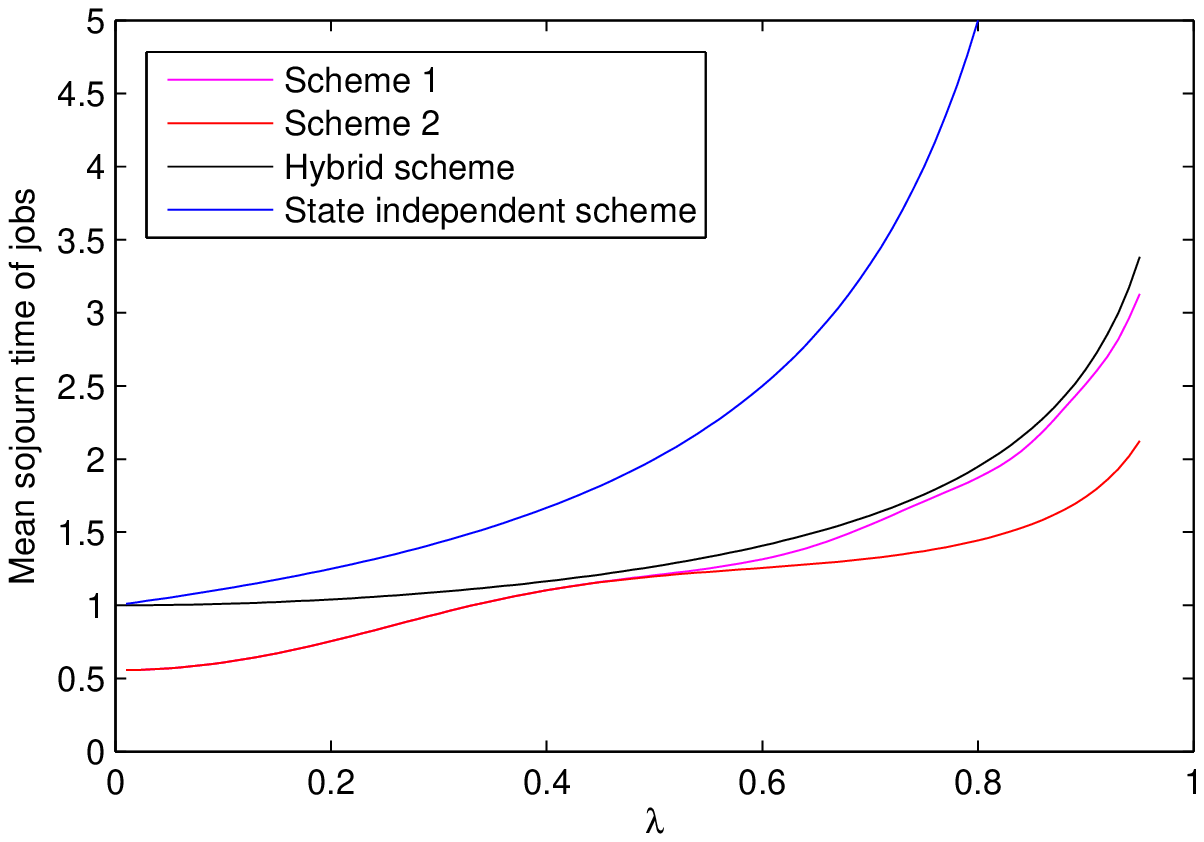}
 \caption{Mean sojourn time jobs as a function of $\lambda$ for different values of $N$.
 We set $M=2$, $C_1=1/5$, $C_2=9/5$,  $\gamma_1=\gamma_2=0.5$, and $d_1=d_2=2$.
 Routing probabilities for the state independent scheme and the hybrod SQ($d$) scheme are not optimized.}
 \label{fig:comparison2}
 \end{figure}

 To observe the effect of fixing the routing probabilities for the
 hybrid SQ($d$) scheme and the state independent scheme, 
 we choose $p_i=\frac{\gamma_i C_i}{\sum_{j\in \cal{J}}\gamma_j C_j}$  
 for each server type $i \in \cal{J}$. This choice of routing probabilities 
 ensures that all arrival rates in the maximal stability region can be supported
 by the system operating under either the state independent scheme
 or the Hybrid SQ($d$) scheme. We choose the same parameter setting as before and plot
 mean sojourn time of jobs as a function of $\lambda$ in Figure~\ref{fig:comparison2} for
 the schemes under consideration.  In this case, we notice that both Scheme~1 and
 Scheme~2 outperform the hybrid SQ($d$) scheme. Hence, in the scenarios where
 estimation of arrival rates is not possible, Scheme~2 is a better choice than the hybrid SQ($d$)
 scheme.

\begin{table}[ht]
 \renewcommand{\arraystretch}{1.5}
 \caption{Insensitivity of Scheme~1}
 \centering
 \begin{tabular}{c c c c}
   \hline
   \bf{$\lambda$} & \begin{tabular}[c]{@{}c@{}} Mean sojourn time $\bar{T}$\\(Theoretical) \end{tabular}  & \begin{tabular}[l]{@{}c@{}}Constant\\(Simulation)\end{tabular} & \begin{tabular}[c]{@{}c@{}}Power Law\\(Simulation)\end{tabular}\\
   \hline
   0.2 & 0.8076 & 0.8106 & 0.8098\\
   \hline
   0.3 & 0.8609 & 0.8642 & 0.8640\\
   \hline
   0.5 & 0.9809 & 0.9852 & 0.9840\\
   \hline
   0.7 & 1.1696 & 1.1759 & 1.1757\\
   \hline
   0.8 & 1.3687 & 1.3741 & 1.3740\\
   \hline
   0.9 & 1.7531 & 1.7641 & 1.7645\\
   \hline
 \end{tabular}
 \label{N_100}
 \end{table}

We now numerically investigate the behaviour of the proposed
schemes under different job length distributions. 
In  Table~\ref{N_100}, mean sojourn
time of jobs under Scheme~1 is shown 
as a function of  $\lambda$, for the following distributions.

\begin{enumerate}
\item {\em Constant}: We consider job length distribution
having the cumulative distribution given by
$F(x)=0$ for $0 \leq x < 1$,
and $F(x)=1$, otherwise.

\item {\em Power law}: We consider job length distribution having
cumulative distribution function given by 
$F(x)=1-1/4x^2$ for $x \geq \frac{1}{2}$ and $F(x)=0$, otherwise.
\end{enumerate}
For both distributions we have $\mu=1$. 
We choose the following parameter values $M=2$, $C_1=4/3$, $C_2=2/3$, $N=100$,
$\gamma_1=\gamma_2=\frac{1}{2}$, and $d_1=d_2=2$.
We observe that there is insignificant change in the mean sojourn
time of jobs when
the job length distribution type is changed.  
The results, therefore, justify the insensitivity property
as discussed in Remark~\ref{rmk:insensitivity}.

\section{Conclusion}
\label{sec:conclusion}

We considered randomized job assignment schemes in  
a multi-server system consisting of $N$ parallel processor sharing servers,
categorized into $M$ ($\ll N$) different types
according to their processing capacity or speed. In the proposed schemes, 
a small number of servers from each type is sampled uniformly at random at each arrival instant. 
It was shown that due to such sampling the schemes achieve  the maximal stability region.
Mean field analysis was carried out to show that asymptotic independence among servers holds even when
$M$ is finite and exchangeability holds only within servers of the same type. 
The existence and uniqueness of stationary solution of the mean field and
doubly exponentially decreasing nature of the tail distribution of the number of jobs was
established. Numerical studies have shown that, when the estimates
of arrival rates are not available, the proposed schemes 
offer simpler alternatives to achieving lower mean sojourn time of jobs.

\appendix

\section{}
\label{proof:uniqueness_sol}

We will prove Proposition~\ref{thm:uniqueness_sol} only for the system~\eqref{eq:diff1_x}-\eqref{eq:diff2_x}.
The proof for the system~\eqref{eq:diff1_xc}-\eqref{eq:diff2_xc} follows similarly.

Define $\theta(x)=[\min(x,1)]_{+}$, where $[z]_{+}=\max\cbrac{0,z}$ and 
let us consider the following modification of~\eqref{eq:diff1_x}-\eqref{eq:diff2_x}:
\begin{align}
 \mf{u}(0) &= \mf{g}, \label{eq:diff1_x'}\\
 \dot{\mf{u}}(t) &= \mf{\hat{l}}(\mf{u}(t)), \label{eq:diff2_x'}
\end{align}
where the mapping $\mf{\hat{l}}:\brac{\mb{R}^{\mb{Z}_+}}^M \to \brac{\mb{R}^{\mb{Z}_+}}^M$
is given by

\begin{align}
\mv{\hat{l}}{0}{j}(\mf{u}) &= 0,  \text{ for } j \in \mathcal{J}, \label{eq:diff3_x'}\\
 \mv{\hat{l}}{k}{j}(\mf{u}) &= \frac{\lambda}{\gamma_{j}} \left[\brac{\theta\left(\mv{u}{k-1}{j}\right)}^{d_{j}} - \brac{\theta\left(\mv{u}{k}{j}\right)}^{d_{j}}\right]_{+}
\prod_{i=1}^{j-1} \left(\theta\brac{\mv{u}{k-1}{i}}\right)^{d_{i}}\label{eq:diff4_x'}\\
\times \prod_{i=j+1}^{M} & \left(\theta\brac{\mv{u}{k}{i}}\right)^{d_{i}} 
- \mu C_{j} \left[\theta\brac{\mv{u}{k}{j}}\! - \!\theta\brac{\mv{u}{k+1}{j}}\right]_{+}, \text{ for } k \geq 1, j \in \mathcal{J}. \nonumber
\end{align}
Clearly, the right hand side of~\eqn{eq:diff4_x} and~\eqn{eq:diff4_x'}
are equal if $\mathbf{u} \in \Ub^M$. Therefore, the two systems must have identical
solutions in $\Ub^M$. Also if $\mf{g} \in \Ub^M$, then any solution of
the modified system remains within $\Ub^M$. 
This is because of the facts that if $u^{(j)}_n(t)=u^{(j)}_{n+1}(t)$ for
some $j$, $n$, $t$, then $\hat{l}^{(j)}_n(\mf{u}(t)) \geq 0$ and 
$\hat{l}^{(j)}_{n+1}(\mf{u}(t)) \leq 0$, and if $\mv{u}{n}{j}(t)=0$
for some $j$, $n$, $t$, then $\mv{\hat{l}}{n}{j}(\mf{u}(t)) \geq 0$.
Hence, to prove the uniqueness of solution of
\eqn{eq:diff1_x}-\eqn{eq:diff2_x}, we need to show
that the modified system~\eqn{eq:diff1_x'}-\eqn{eq:diff2_x'}
has a unique solution in $(\mb{R}^{\mb{Z}_{+}})^{M}$.
We now extend the distance metric defined in~\eqref{eq:norm} to the space $(\mb{R}^{\mb{Z}_{+}})^{M}$.

Using the metric defined in~\eqn{eq:norm} and the facts
that $\abs{x_+-y_+} \leq \abs{x-y}$ for any $x,y \in \mb{R}$,
$\abs{a_1b_1^m-a_2b_2^m} \leq \abs{a_1-a_2}+m\abs{b_1-b_2}$ for any 
$a_1,a_2,b_1,b_2 \in [0,1]$, and $\abs{\theta(x)-\theta(y)} \leq \abs{x-y}$
for any $x,y \in \mb{R}$ we obtain

\begin{align}
\norm{\mf{\hat{l}}(\mf{u})} &\leq K_1, \label{eq:bound}\\
\norm{\mf{\hat{l}}(\mf{u})-\mf{\hat{l}}(\mf{w})} &\leq K_2 \norm{\mf{u}-\mf{w}},
\label{eq:lip}
\end{align}
where $\mf{u}, \mf{w} \in (\mb{R}^{\mb{Z}_{+}})^{M}$, $K_1$ and $K_2$ are constants defined as
$K_1=\frac{\lambda}{\min_{j \in \cal{J}} \gamma_j} +\mu (\max_{ j \in \cal{J}} C_j)$
and $K_2= 4M \lambda \frac{\max_{j \in\cal{J}} d_j}{\min_{j \in \cal{J}} \gamma_j}+ 3 \mu (\max_{1 \leq j \leq M} C_j)$.
The uniqueness now follows from inequalities~\eqn{eq:bound}
and~\eqn{eq:lip} by using Picard's iteration
technique since $(\mb{R}^{\mb{Z}_{+}})^{M}$ is complete under the metric defined in~\eqn{eq:norm}. \qed

\section{}
\label{proof:convergence_semigroup}

We prove Proposition~\ref{thm:convergence_semigroup}
by showing that the generators of the corresponding
semigroups converge as $N \to \infty$.
We first recollect the following from~\cite{Ethier_Kurtz_book}.

\begin{itemize}
\item The generator
$\mathbf{A}_N$ of the semigroup $\cbrac{\mf{T}_N(t)}_{t \geq 0}$ acting on functions 
$f:\prod_{j=1}^{M}\UN \rightarrow \mathbb{R}$ is
given by $\mathbf{A}_N f(\mathbf{g})=\sum_{\mf{h} \neq \mf{g}} q_{\mf{g}\mf{h}} \brac{f(\mf{h})-f(\mf{g})}$, 
where $q_{\mf{g}\mf{h}}$, 
with $\mf{g},\mf{h} \in  \prod_{j=1}^{M} \UN$, denotes the transition
rate from state $\mf{g}$ to state $\mf{h}$.

\item The generator
$\mathbf{A}$ of the semigroup $\cbrac{\mf{T}(t)}_{t \geq 0}$ acting on functions 
$f:\Ub^M \rightarrow \mathbb{R}$ having bounded partial derivatives is
given by $\mf{A}f(\mf{g})=\lim_{t \downarrow 0} \frac{\mf{T}(t)f(\mf{g})-f(\mf{g})}{t}=\frac{d}{dt} f(\mf{u}(t,\mf{g}))\rvert_{t=0}$.
\end{itemize}

In the following lemma, we characterize the
the generator $\mf{A}_N$ associated with the process $\mf{x}_N(t)$.
  
\begin{lemma}
\label{thm:generator} 
Let $\mathbf{g} \in \prod_{j=1}^{M}\UN$ be any state of the process $\mf{x}_N(t)$
and {$\mathbf{e}(n,j)=\brac{e_k^{(i)}}_{k \in \mb{Z}_+, i \in \cal{J}}$}
be the unit vector with $e^{(j)}_n=1$ and $e^{(i)}_k=0$
if $i \neq j$ and $k \neq n$. 
Under Scheme~1, the generator
$\mathbf{A}_N$ of the Markov process $\mf{x}_N(t)$
acting on functions $f:\prod_{j=1}^{M}\UN \rightarrow \mathbb{R}$ is
given by

\begin{multline}
\mathbf{A}_N f(\mathbf{g})=  N \lambda  \sum_{j=1}^{M} \sum_{n \geq 1}
\sbrac{\brac{\mv{g}{n-1}{j}}^{d_j}-\brac{\mv{g}{n}{j}}^{d_j}}\prod_{i=1}^{j-1}\brac{\mv{g}{n-1}{i}}^{d_i}
\\ \times \prod_{i=j+1}^M \brac{\mv{g}{n}{i}}^{d_i}
\sbrac{f(\mathbf{g}+\frac{\mathbf{e}(n,j)}{N \gamma_j})-f(\mathbf{g})} \\
+ \mu N \sum_{n \geq 1} \sum_{j=1}^{M} \gamma_j C_j \sbrac{g^{(j)}_n-g^{(j)}_{n+1}}
\times \sbrac{f(\mathbf{g}-\frac{\mathbf{e}(n,j)}{N \gamma_j})-f(\mathbf{g})}.\label{eq:gen_x}   
\end{multline}
Under Scheme~2, the generator
$\mathbf{A}_N$ of the Markov process $\mf{x}_N(t)$
acting on functions $f:\prod_{j=1}^{M}\UN \rightarrow \mathbb{R}$ is
given by

\begin{multline}
\mathbf{A}_N f(\mathbf{g})=  N \lambda  \sum_{j=1}^{M} \sum_{n \geq 1}
\sbrac{\brac{\mv{g}{n-1}{j}}^{d_j}-\brac{\mv{g}{n}{j}}^{d_j}}\prod_{i=1}^{j-1}\brac{\mv{g}{\ktil{n-1}{j}{i}}{i}}^{d_i}\\
\times \prod_{i=j+1}^M \brac{\mv{g}{\ksub{n-1}{j}{i}}{i}}^{d_i}
 \sbrac{f(\mathbf{g}+\frac{\mathbf{e}(n,j)}{N \gamma_j})-f(\mathbf{g})} \\
+ \mu N \sum_{n \geq 1} \sum_{j=1}^{M} \gamma_j C_j \sbrac{g^{(j)}_n-g^{(j)}_{n+1}}
\times \sbrac{f(\mathbf{g}-\frac{\mathbf{e}(n,j)}{N \gamma_j})-f(\mathbf{g})}.\label{eq:gen_xc}   
\end{multline}
\end{lemma}

\begin{proof}
We only prove the lemma for Scheme~1. For Scheme~2, it can be shown
similarly.

We first consider an arrival joining a server of type $j$ with exactly
$n-1$ unfinished jobs, when the state of the system is $\mf{g}$.
This corresponds to the transition from state $\mf{g}$ to the state
$\mf{g}+\frac{\mf{e}(n,j)}{N \gamma_j}$.
The term $\left(\left(\mv{g}{n-1}{j}\right)^{d_{j}} - \left(\mv{g}{n}{j}\right)^{d_{j}}\right)$  
$\times\prod_{i=1}^{j-1} \left(\mv{g}{n-1}{i}\right)^{d_{i}}
\prod_{i=j+1}^{M} \left(\mv{g}{n}{i}\right)^{d_{i}}$ 
denotes the probability with which an arrival joins a type $j$ server with exactly $n-1$ jobs.
This is because a job joins a server of type $j$ with exactly $n-1$
occupancy only when the following conditions are satisfied:

\begin{itemize}
\item Among the $d_j$ sampled servers of type $j$, at least one has
exactly $n-1$ jobs and the rest of them have at least $n$ jobs.

\item For each $i < j$, all the $d_i$ sampled servers of type $i$ have at least
$n-1$ jobs.

\item  For each $i > j$, all the $d_i$ servers of type $i$ have at least $n$
jobs.
\end{itemize} 
Since the arrival rate of jobs is $N \lambda$, the rate of the above transition
is given by

\begin{equation}
q_{\mf{g}, \mf{g}+\frac{\mf{e}(n,j)}{N\gamma_j}}=N \lambda \sbrac{\brac{\mv{g}{n-1}{j}}^{d_j}-\brac{\mv{g}{n}{j}}^{d_j}}\prod_{i=1}^{j-1}\brac{\mv{g}{n-1}{i}}^{d_i}
\prod_{i=j+1}^M \brac{\mv{g}{n}{i}}^{d_i}
\end{equation}
Further, the rate at which jobs depart from a server of type $j$ having exactly $n$
jobs is $\mu C_j N \gamma_j \left(\mv{g}{n}{j} - \mv{g}{n+1}{j}\right)$. 
The expression~\eqref{eq:gen_x}
now follows directly from the definition of $\mf{A}_N$.
\end{proof}

We now show  that the solutions
$\mf{u}(t,\mf{g})$ of~\eqref{eq:diff1_x}-\eqref{eq:diff2_x}
and~\eqref{eq:diff1_xc}-\eqref{eq:diff2_xc} are smooth with respect
to the initial point $\mf{g}$ and their partial derivatives are bounded.

\begin{lemma}
For each $j$, $n$, $j^\prime$, $n'$, $i$, $k$, and $t \geq 0$, the partial derivatives
$\frac{\partial\mf{u}(t,\mf{g})}{\partial \mv{g}{n}{j}}$,
$\frac{\partial^2\mf{u}(t,\mf{g})}{\partial {\mv{g}{n}{j}}^2}$,
and $\frac{\partial^2 \mf{u}(t,\mf{g})}{\partial \mv{g}{n}{j} \partial \mv{g}{n'}{j'}}$
exist for $\mf{g} \in \Ub^M$ and satisfy

\begin{equation}
\abs{\frac{\partial \mv{u}{k}{i}(t,\mf{g})}{\partial \mv{g}{n}{j}}} \leq \exp(B_1 t)
\label{eq:partial_bound1}
\end{equation}
and
\begin{equation}
\abs{\frac{\partial^2 \mv{u}{k}{i}(t,\mf{g})}{\partial {\mv{g}{n}{j}}^2}},
\abs{\frac{\partial^2 \mv{u}{k}{i}(t,\mf{g})}{\partial \mv{g}{n}{j} \partial \mv{g}{n'}{j'}}} \leq \frac{B_2}{B_1}(\exp(2B_1 t)
-\exp(B_1 t)),
\label{eq:partial_bound2}
\end{equation}
where $B_1=\frac{2\lambda\sum_{j \in \cal{J}}d_j}{\min_{j \in \cal{J}}\gamma_j}+2\mu\brac{\max_{j \in \cal{J}}C_j}$,
and $B_2=\frac{2\lambda\brac{\sum_{j \in \cal{J}}d_j}^2}{\min_{j \in \cal{J}}\gamma_j}$.
\end{lemma}

\begin{proof}
The proof follows the same line of arguments as the proof of Lemma~3.2 of~\cite{Martin_AAP_1999}.
We omit the details.
%
%
\end{proof}

\subsection*{Proof of Proposition~\ref{thm:convergence_semigroup}}
The proof is essentially the same as the proof Theorem~2 of~\cite{Martin_AAP_1999}.
We omit the details.\qed

\section{}
\label{proof:existence_fixedp}

We prove the existence of equilibrium point for Scheme~1. Similar arguments apply for Scheme~2.
For simplicity of exposition, we further restrict ourselves to the $M=2$ case.
However, the proof can be extended to any $M \geq 2$.

The idea is to construct sequences $\cbrac{\mv{P}{k}{j}, k \in \mb{Z}_+}$ for $j=1,2$
such that they satisfy the following three properties
\begin{enumerate}[P.1]
\item Equation~\eqref{eq:tailhet_x} for $j=1,2$. \label{prop:p1}
\item $\mv{P}{k}{j} \geq \mv{P}{k+1}{j} \geq 0$ for all $k \in \mb{Z}_+$, $j=1,2$.
\item $\mv{P}{k}{j} \to 0$ as $k \to \infty$ for $j=1,2$.
\end{enumerate}
According to Proposition~\ref{thm:tail_het_properties}, we see that
$\mf{P}=\cbrac{\mv{P}{k}{j}, k\in \mb{Z}_+,j\in \cbrac{1,2}}$ 
with components $\mv{P}{k}{j}$ satisfying the above properties, 
must be an equilibrium point of the system~\eqref{eq:diff1_x}-\eqref{eq:diff2_x}
and also must lie in the space $\U^2$.
Note that if (P.1) holds and $\mv{P}{k}{j} \geq 0$ for all $k$ and $j$,
then $\mv{P}{k}{j} \geq \mv{P}{k+1}{j}$.
 
We now construct the sequences $\cbrac{\mv{P}{l}{1}(\alpha), l\in \mb{Z}_+}$
and $\cbrac{\mv{P}{l}{2}(\alpha), l\in \mb{Z}_+}$ as functions of the
real variable $\alpha$ as follows: 

\begin{align}
\mv{P}{0}{1}(\alpha) & =1.\\
\mv{P}{0}{2}(\alpha) &=1.\\
\mv{P}{1}{1}(\alpha) &=\alpha. \label{eq:init1}\\
\mv{P}{1}{2}(\alpha) &=\Delta_2\brac{1-\frac{\alpha}{\Delta_1}}. \label{eq:init2}\\
\mv{P}{l+2}{1}(\alpha) &= \mv{P}{l+1}{1}(\alpha)-\Delta_1 \left(\left(\mv{P}{l}{1}(\alpha)\right)^{d_{1}} - \left(\mv{P}{l+1}{1}(\alpha)\right)^{d_{1}}\right)\label{eq:rec1}\\
& \hspace{3cm}\times \left(\mv{P}{l+1}{2}(\alpha)\right)^{d_{2}}, l \geq 0 \nonumber \\
\mv{P}{l+2}{2}(\alpha) &= \mv{P}{l+1}{2}(\alpha)-\Delta_2 \left(\left(\mv{P}{l}{2}(\alpha)\right)^{d_{2}} - \left(\mv{P}{l+1}{2}(\alpha)\right)^{d_{2}}\right) \label{eq:rec2}\\
& \hspace{3cm}\times \left(\mv{P}{l}{1}(\alpha)\right)^{d_{1}},  l \geq 0 \nonumber
\end{align}
Combining the above relations we obtain

\begin{equation}
\sum_{j =1}^2 \frac{\mv{P}{l+1}{j}(\alpha)}{\Delta_j}= \prod_{j=1}^2 \brac{\mv{P}{l}{j}(\alpha)}^{d_j}, \text{ for } l \geq 0
\label{eq:couple}
\end{equation}
Note that that
the sequences $\cbrac{\mv{P}{l}{1}(\alpha), l\in \mb{Z}_+}$
and $\cbrac{\mv{P}{l}{2}(\alpha), l\in \mb{Z}_+}$ are constructed such that they
satisfy property (P.1). Hence, the
the proof will be complete if for some $\alpha \in (0,1)$ the properties
(P.2) and (P.3) are satisfied. 
We first proceed to find $\alpha \in (0,1)$ such that
the sequences $\cbrac{\mv{P}{l}{1}(\alpha), l\in \mb{Z}_+}$
and $\cbrac{\mv{P}{l}{2}(\alpha), l\in \mb{Z}_+}$ are both positive sequences
of real numbers in $[0,1]$. This will ensure that (P.2) is satisfied.


Note that $\mv{P}{l}{1}(1)=1$ for all $l \in \mb{Z}_+$.
Hence, from~\eqref{eq:init2} we have 
$\mv{P}{1}{2}(1)=\Delta_2\brac{1-\frac{1}{\Delta_1}}$ and from~\eqref{eq:rec2} we have

\begin{equation}
\mv{P}{l+2}{2}(1) = \mv{P}{l+1}{2}(1)-\Delta_2 \left(\left(\mv{P}{l}{2}(1)\right)^{d_{2}} - 
\left(\mv{P}{l+1}{2}(1)\right)^{d_{2}}\right) \text{ for } l \geq 0
\label{eq:rec2_1}
\end{equation}
Notice that the stability condition~\eqref{eq:maximal_stability} reduces to 

\begin{equation}
\frac{1}{\Delta_1}+\frac{1}{\Delta_2} > 1,
\end{equation}
which implies that $\mv{P}{1}{2}(1) < 1$.
We claim that there exists some $l \geq 1$ such that $\mv{P}{l}{2}(1) < 0$.
Let us assume this is not true. Therefore, $\mv{P}{l}{2}(1) \geq 0$ for all $l\geq 0$.
By~\eqref{eq:rec2_1}, this implies that $\cbrac{\mv{P}{l}{2}(1), l \geq 0}$ is a 
non-decreasing sequence of numbers in $[0,1)$. 
Hence by monotone convergence theorem $\lim_{l \to \infty}\mv{P}{l}{2}(1)$ exists. 
Let this limit be denoted by $\beta$, where $0 \leq \beta < 1$. Thus, adding~\eqref{eq:rec2_1} 
for $l \geq 0$ and using $\lim_{l \to \infty}\mv{P}{l}{2}(1)=\beta$ we obtain
 
\begin{align*}
\brac{1-\frac{1}{\Delta_1}} &= \frac{\beta}{\Delta_2}+1-\beta^{d_2}\\
                                               &> \beta  \brac{1-\frac{1}{\Delta_1}} +1 - \beta^{d_2}.
\end{align*}
Hence, $\brac{1-\frac{1}{\Delta_1}} > \frac{1-\beta^{d_2}}{1-\beta} \geq 1$.
This is a contradiction since $\Delta_1 > 0$. 
Hence, there exists $l \geq 1$ such that $\mv{P}{l}{2}(1) < 0$.

Observe that $\mv{P}{l}{2}\brac{\Delta_1\brac{1-\frac{1}{\Delta2}}}=1$ for all $l \geq 0$.
Hence, with same line of arguments as above, it can be shown that there exists 
$l \geq 1$ such that $\mv{P}{l}{1}\brac{\Delta_1\brac{1-\frac{1}{\Delta2}}} < 0$.

Now from~\eqref{eq:init2} and~\eqref{eq:rec2} it is easily seen that
$\mv{P}{l}{2}(0) >0$ for all $l \geq 0$. From the same relations
we also observe that $\mv{P}{l}{2}\brac{\Delta_1\brac{1-\frac{1}{\Delta2}}}=1 > 0$
for all $l \geq 0$. Combining the two we have

\begin{equation}
\mv{P}{l}{2}\brac{\max\brac{0,\Delta_1\brac{1-\frac{1}{\Delta_2}}}} > 0
\end{equation}
Further, observe that $\mv{P}{1}{2}(\Delta_1) = 0$. Hence, there must exist
at least one root of $\mv{P}{1}{2}(\alpha)$ in the following range

\begin{equation}
\alpha \in \left(\max\brac{0,\Delta_1\brac{1-\frac{1}{\Delta_2}}},\Delta_1\right. \left.\vphantom{\frac{1}{\Delta_1}}\right].
\end{equation} 
Let $\mv{r}{1}{2}$ denote the minimum root of $\mv{P}{1}{2}(\alpha)$ in
the above range. Therefore, in the range

\begin{equation}
\alpha \in \left(\max\brac{0,\Delta_1\brac{1-\frac{1}{\Delta_2}}},\min\brac{1,\mv{r}{1}{2}}\right. \left.\vphantom{\frac{1}{\Delta_1}}\right],
\label{eq:range}
\end{equation}
we must have $\mv{P}{1}{2}(\alpha) \geq 0$. (Note that
the right limit can be combined with $1$ because
of the minimality of $\mv{r}{1}{2}$). Putting
$l=0$, $\alpha=\mv{r}{1}{2}$ in~\eqref{eq:rec2}
we observe that $\mv{P}{2}{2}\brac{\mv{r}{1}{2}} < 0$.
Hence, using the same line arguments we conclude
that in the range

\begin{equation}
\alpha \in \left(\max\brac{0,\Delta_1\brac{1-\frac{1}{\Delta_2}}},\min\brac{1,\mv{r}{2}{2}}\right. \left.\vphantom{\frac{1}{\Delta_1}}\right],
\label{eq:range1}
\end{equation} 
both $\mv{P}{1}{2}(\alpha), \mv{P}{2}{2}(\alpha) \geq 0$,
where $\mv{r}{2}{2}$ denotes the minimum root
of $\mv{P}{2}{2}(\alpha)$ in the range defined in~\eqref{eq:range}. 
Therefore by~\eqref{eq:rec2} we also have
$\mv{P}{1}{2}(\alpha) \geq \mv{P}{2}{2}(\alpha) >0$ in the above range.
Repeating the same argument again for $\mv{P}{3}{2}(\alpha)$
we find that $\mv{P}{1}{2}(\alpha) \geq \mv{P}{2}{2}(\alpha) \geq \mv{P}{3}{2}(\alpha) \geq 0$
holds in the range

\begin{equation}
\alpha \in \left(\max\brac{0,\Delta_1\brac{1-\frac{1}{\Delta_2}}},\min\brac{1,\mv{r}{3}{2}}\right. \left.\vphantom{\frac{1}{\Delta_1}}\right],
\label{eq:range2}
\end{equation} 
where $\mv{r}{3}{2}$ denotes the minimum root
of $\mv{P}{3}{2}(\alpha)$ in the range defined in~\eqref{eq:range1}.

Trivially, we have $\mv{P}{1}{1}(\alpha) > 0$
in the range defined in~\eqref{eq:range2}.
Now from~\eqref{eq:rec1} we have $\mv{P}{2}{1}(0)=-\Delta_1 \Delta_2 ^{d_2} < 0$. 
Also, from definition of $\mv{r}{3}{2}$ we know that
$\mv{P}{3}{2}(\mv{r}{3}{2}) =0$. Now, by putting $\alpha=\mv{r}{3}{2}$
and $l=1$ in~\eqref{eq:rec2} we obtain

\begin{align*}
\mv{P}{2}{2}(\mv{r}{3}{2}) &= \Delta_2 \sbrac{\brac{\mv{P}{1}{2}(\mv{r}{3}{2})}^{d_2}-\brac{\mv{P}{2}{2}(\mv{r}{3}{2})}^{d_2}} \brac{\mv{r}{3}{2}}^{d_1}\\
& \leq \Delta_2 \brac{\mv{P}{1}{2}(\mv{r}{3}{2})}^{d_2}\brac{\mv{r}{3}{2}}^{d_1} (\text{since } \mv{P}{2}{2}(\mv{r}{3}{2}) \geq 0)
\end{align*}
Again, by putting $l=2$ and $\alpha=\mv{r}{3}{2}$ in~\eqref{eq:couple} and using the above we obtain
$\mv{P}{2}{1}(\mv{r}{3}{2}) \geq 0$. Therefore, there exists at least one root of $\mv{P}{2}{1}(\alpha)$
in the interval $\left(\right. 0,\mv{r}{3}{2}\left.\right]$.
Denote the maximum of all such roots to be $\mv{r}{2}{1}$. Hence, in the range

\begin{equation}
\alpha \in \left[\max\brac{\mv{r}{2}{1},\Delta_1\brac{1-\frac{1}{\Delta_2}}},\min\brac{1,\mv{r}{3}{2}}\right. \left.\vphantom{\frac{1}{\Delta_1}}\right],
\end{equation} 
we have $\mv{P}{1}{1}(\alpha)
\geq \mv{P}{2}{1}(\alpha) \geq 0 $ along with $\mv{P}{1}{2}(\alpha)
\geq \mv{P}{2}{2}(\alpha) \geq \mv{P}{3}{2}(\alpha) \geq 0$. Again from~\eqref{eq:rec1}
we observe that $\mv{P}{3}{1}(\mv{r}{2}{1}) < 0$. Further, putting $l=3$ and  $\alpha=\mv{r}{3}{2}$ in~\eqref{eq:couple}
we obtain $\mv{P}{3}{1}(\mv{r}{3}{2}) \geq 0$. 
Thus, there must be at least one root of $\mv{P}{3}{1}(\alpha)$
in the range $\left(\right.\mv{r}{2}{1}, \mv{r}{3}{2}\left. \right]$.
Let $\mv{r}{3}{1}$ denote the maximum
root in the interval.
Hence, in the interval

\begin{equation}
\alpha \in \left[\max\brac{\mv{r}{3}{1},\Delta_1\brac{1-\frac{1}{\Delta_2}}},\min\brac{1,\mv{r}{3}{2}}\right. \left.\vphantom{\frac{1}{\Delta_1}}\right],
\end{equation}
we have $\mv{P}{1}{1}(\alpha)
\geq \mv{P}{2}{1}(\alpha) \geq \mv{P}{3}{1}(\alpha) \geq 0 $ along with $\mv{P}{1}{2}(\alpha)
\geq \mv{P}{2}{2}(\alpha) \geq \mv{P}{3}{2}(\alpha) \geq 0$.
Similarly, from~\eqref{eq:rec1} we have $\mv{P}{4}{1}(\mv{r}{3}{1}) < 0$
and from~\eqref{eq:rec2} we have $\mv{P}{4}{1}(\mv{r}{3}{2}) \geq 0$.
Thus, there must be at least one root of $\mv{P}{4}{1}(\alpha)$
in the range $\left(\right.\mv{r}{3}{1}, \mv{r}{3}{2}\left. \right]$.
Denote the maximum of all such roots by
$\mv{r}{4}{1}$. 
Hence, in the interval

\begin{equation}
\alpha \in \left[\max\brac{\mv{r}{4}{1},\Delta_1\brac{1-\frac{1}{\Delta_2}}},\min\brac{1,\mv{r}{3}{2}}\right. \left.\vphantom{\frac{1}{\Delta_1}}\right],
\end{equation}
we have $\mv{P}{1}{1}(\alpha)
\geq \mv{P}{2}{1}(\alpha) \geq \mv{P}{3}{1}(\alpha) \geq \mv{P}{4}{1}(\alpha) \geq 0 $ 
and $\mv{P}{1}{2}(\alpha)
\geq \mv{P}{2}{2}(\alpha) \geq \mv{P}{3}{2}(\alpha) \geq 0$.

Using the same line of arguments as above the following inductive
hypothesis can be proved: If, for $k \geq 0$,
$\mv{P}{1}{1}(\alpha)\geq \mv{P}{2}{1}(\alpha)\ldots \geq \mv{P}{4+3k}{1}(\alpha) \geq 0$
and $\mv{P}{1}{2}(\alpha)\geq \mv{P}{2}{2}(\alpha)\ldots \geq \mv{P}{3+3k}{1}(\alpha) \geq 0$
hold in the range

\begin{equation}
\alpha \in \left[\max\brac{\mv{r}{4+3k}{1},\Delta_1\brac{1-\frac{1}{\Delta_2}}},\min\brac{1,\mv{r}{3+3k}{2}}\right. \left.\vphantom{\frac{1}{\Delta_1}}\right],
\label{eq:range_induc1}
\end{equation}
then $\mv{P}{1}{1}(\alpha)\geq \mv{P}{2}{1}(\alpha)\ldots \geq \mv{P}{4+3(k+1)}{1}(\alpha) \geq 0$
and $\mv{P}{1}{2}(\alpha)\geq \mv{P}{2}{2}(\alpha)\ldots \geq \mv{P}{3+3(k+1)}{1}(\alpha) \geq 0$
hold in the range

\begin{equation}
\alpha \in \left[\max\brac{\mv{r}{4+3(k+1)}{1},\Delta_1\brac{1-\frac{1}{\Delta_2}}},\min\brac{1,\mv{r}{3+3(k+1)}{2}}\right. \left.\vphantom{\frac{1}{\Delta_1}}\right],
\label{eq:range_induc2}
\end{equation}
and the interval in~\eqref{eq:range_induc2} is included in the interval in~\eqref{eq:range_induc1}.

The decreasing sequence of compact intervals

\begin{equation}
\left[\max\brac{\mv{r}{4+3k}{1},\Delta_1\brac{1-\frac{1}{\Delta_2}}},\min\brac{1,\mv{r}{3+3k}{2}}\right. \left.\vphantom{\frac{1}{\Delta_1}}\right], \text{ for } k \geq 0
\end{equation}
eventually become strict subsets of the interval $[0,1]$ as discussed in the beginning. 
Further, the intersection of all such compact intervals must be non-empty
due to the Cantor's intersection theorem. Hence, we have shown that there exists $\alpha \in (0,1)$
such that the sequences $\cbrac{\mv{P}{l}{1}(\alpha), l\in \mb{Z}_+}$
and $\cbrac{\mv{P}{l}{2}(\alpha), l\in \mb{Z}_+}$ are both positive non-increasing sequences
of real numbers in $[0,1]$.

We now proceed to show that the above sequences satisfy property (P.3).
Let $\lim_{l \to \infty} \mv{P}{l}{1}(\alpha)=\xi_1 \geq 0$ and $\lim_{l \to \infty} \mv{P}{l}{2}(\alpha)=\xi_2\geq 0$,
where $\alpha$ is chosen such that both sequences
become positive and non-increasing. Now, taking limit of~\eqref{eq:couple}
as $l \to \infty$ we have

\begin{equation}
\sum_{j =1}^2 \frac{\xi_{j}}{\Delta_j}= \prod_{j=1}^2 \brac{\xi_{j}}^{d_j}.
\label{eq:limit_couple}
\end{equation}
Now using the stability criterion and the fact that $0 \leq \xi_1,\xi_2 \leq 1$ we have

\begin{align*}
&\frac{1}{\Delta_1}+\frac{1}{\Delta_2} > 1\\
\Rightarrow &\frac{\xi_2}{\Delta_1}+\frac{\xi_2}{\Delta_2} \geq \xi_2 \geq \xi_2^{d_2}
\end{align*}
with equality holding if and only if $\xi_2=0$.
Further, we have

\begin{align*}
&\frac{1}{\Delta_1}+\frac{\xi_2}{\Delta_2} \geq \frac{\xi_2}{\Delta_1}+\frac{\xi_2}{\Delta_2} \geq \xi_2^{d_2}
\end{align*}
Hence, by multiplying both sides with $\xi_1$ we have

\begin{align*}
&\frac{\xi_1}{\Delta_1}+\frac{\xi_1 \xi_2}{\Delta_2} \geq \xi_1 \xi_2^{d_2} \geq \xi_1^{d_1} \xi_2^{d_2},
\end{align*}
with equality if and only if $\xi_1=\xi_2=0$. Again, since
$\xi_1 \leq 1$ we have

\begin{align*}
&\frac{\xi_1}{\Delta_1}+\frac{\xi_2}{\Delta_2} \geq \frac{\xi_1}{\Delta_1}+\frac{\xi_1 \xi_2}{\Delta_2} \geq \xi_1 \xi_2^{d_2} \geq \xi_1^{d_1} \xi_2^{d_2},
\end{align*}
Hence, we have shown
\begin{equation}
\frac{\xi_1}{\Delta_1}+\frac{\xi_2}{\Delta_2} \geq \xi_1^{d_1} \xi_2^{d_2}
\end{equation}
with equality holding if and only if $\xi_1=\xi_2=0$. Hence, for~\eqref{eq:limit_couple}
to hold we must have $\xi_1=\xi_2=0$. This proves (P.3) and thus completes the proof. \qed

\section{}
\label{proof:uniqueness_fixedp}
To prove Theorem~\ref{thm:uniqueness_fixedp}, 
we first state the following lemma.
We will write $\mf{g} \leq \mf{g'}$ to mean that $\mv{g}{n}{j} \leq \mv{g'}{n}{j}$
holds for all $n \in \mb{Z}_+$ and $j \in \cal{J}$.

\begin{lemma}
\label{thm:monotonicity}
If $\mf{g} \leq \mf{g'}$ holds, for $\mf{g},\mf{g'} \in \Ub^M$, 
then $\mf{u}(t,\mf{g}) \leq \mf{u}(t,\mf{g'})$ holds for all $t \geq 0$.
\end{lemma}

\begin{proof}
The proof is essentially the same as that of Lemma 3.3 of~\cite{Martin_AAP_1999} and hence omitted.
\end{proof}

We define $\mv{v}{n}{j}(t,\mf{g})=\sum_{k \geq n}\mv{u}{k}{j}(t,\mf{g})$ and
$v_n(t,\mf{g})=\sum_{j \in \cal{J}} \gamma_j \mv{v}{n}{j}(t,\mf{g})$
for each $n \geq 1$ and $j \in \cal{J}$. Further, $\mv{v}{n}{j}(\mf{g})=\sum_{k \geq n}\mv{g}{k}{j}$ and
$v_n(\mf{g})=\sum_{j \in \cal{J}} \gamma_j \mv{v}{n}{j}(\mf{g})$
for each $n \geq 1$ and $j \in \cal{J}$.

\begin{lemma}
If $\mf{g} \in \U^M$, then $\mf{u}(t,\mf{g}) \in \U^M$ for all $t \geq 0$
and

\begin{equation}
\frac{d v_n(t,\mf{g})}{dt}=\lambda\brac{\prod_{j=1}^{M}\brac{\mv{u}{n-1}{j}(t,\mf{g})}^{d_j}-\sum_{j=1}^{M} \frac{\mv{u}{n}{j}(t,\mf{g})}{\Delta_j}} \text{ for all } n \geq 1.
\label{eq:vn}
\end{equation}
In particular,

\begin{equation}
\frac{d v_1(t,\mf{g})}{dt}=\lambda\brac{1-\sum_{j=1}^{M} \frac{\mv{u}{1}{j}(t,\mf{g})}{\Delta_j}}
\label{eq:v1}
\end{equation}
\end{lemma}

\begin{proof}
Suppose that
$\mf{u}(t,\mf{g}) \in \U^M$ holds for all $t \leq \tau$.
Hence, $v_1(\tau,\mf{g}) < \infty$ and $\lim_{n \to \infty} \mv{u}{n}{j}(\tau,\mf{g})=0$ 
for each $j \in \cal{J}$. Summing~\eqref{eq:diff3_x} first over all $k \geq n$
and then over all $j \in \cal{J}$ yields

\begin{equation}
\left.\frac{d v_n(t,\mf{g})}{dt}\right\rvert_{t=\tau}=\lambda\brac{\prod_{j=1}^{M}\brac{\mv{u}{n-1}{j}(\tau,\mf{g})}^{d_j}-\sum_{j=1}^{M} \frac{\mv{u}{n}{j}(\tau,\mf{g})}{\Delta_j}} < \infty,
\end{equation}
for all $n \geq 1$. Hence, for all sufficiently small $h >0$,
we have
$v_n(\tau+h,\mf{g}) < \infty$ for all $n \geq 1$.
This implies that $\mf{u}(\tau+h,\mf{g}) \in \U^M$
for all sufficiently small $h > 0$. This fact
along with $\mf{g}=\mf{u}(0,\mf{g}) \in \U^M$ implies that
$\mf{u}(t,\mf{g}) \in \U^M$ for all $t \geq 0$.
Further,~\eqref{eq:vn} can be obtained by summing~\eqref{eq:diff3_x} 
first over all $k \geq n$
and then over all $j \in \cal{J}$
\end{proof}

\subsection*{Proof of Theorem~\ref{thm:uniqueness_fixedp}}
Clearly, Lemma~\ref{thm:monotonicity} implies the following

\begin{equation}
\mf{u}(t,\min(\mf{g},\mf{P})) \leq \mf{u}(t,\mf{g}) \leq \mf{u}(t,\max(\mf{g},\mf{P}))
\end{equation}
Hence, to prove~\eqref{eq:exp_conv1}, it is sufficient
to show that the convergence holds for $\mf{g} \geq \mf{P}$
and for $\mf{g} \leq \mf{P}$.

We first need to check that for each such $\mf{g}$,
the quantity $v_1(t,\mf{g})$ (and hence also $v_n(t,\mf{g})$ for $n >1$)
is bounded uniformly in $t$. If $\mf{g} \leq \mf{P}$, then by Lemma~\ref{thm:monotonicity}
we have $\mf{u}(t, \mf{g}) \leq \mf{u}(t, \mf{P})=\mf{P}$ for all $t \geq 0$.
Hence, $v_1(t,\mf{g}) \leq v_1(\mf{P})$.

On the other hand, if $\mf{g} \geq \mf{P}$, then
by Lemma~\ref{thm:monotonicity} $\mf{u}(t, \mf{g}) \geq \mf{u}(t, \mf{P})=\mf{P}$.
Hence, we have

\begin{equation}
\sum_{j=1}^{M} \frac{\mv{u}{1}{j}(t,\mf{g})}{\Delta_j} \geq \sum_{j=1}^{M} \frac{\mv{P}{1}{j}}{\Delta_j} = 1
\end{equation}
Thus, from~\eqref{eq:v1} we have $\frac{dv_1(t,\mf{g})}{dt} \leq 0$.
Hence, we have $0 \leq v_1(t,\mf{g}) \leq v_1(\mf{g})$ for all $t \geq 0$.

Since the derivative of $\mv{u}{n}{j}(t)$ is bounded for all $j \in \cal{J}$, 
the convergence $\mf{u}(t,\mf{g}) \to \mf{P}$
will follow from 

\begin{equation}
\int_{0}^{\infty} \brac{\mv{u}{n}{j}(t,\mf{g})-\mv{P}{n}{j}} dt < \infty, \text{ } j \in \mathcal{J}, n \geq 1
\label{eq:more}
\end{equation}
in the case $\mf{g} \geq \mf{P}$, and from

\begin{equation}
\int_{0}^{\infty} \brac{\mv{P}{n}{j}-\mv{u}{n}{j}(t,\mf{g})} dt < \infty, \text{ } j \in \mathcal{J}, n \geq 1
\label{eq:less}
\end{equation}
in the case $\mf{g} \leq \mf{P}$. Both the bounds can be shown similarly.
We discuss the proof of~\eqref{eq:more}.

To prove~\eqref{eq:more} it is sufficient to show that

\begin{equation}
\int_{0}^{\infty} \sum_{j=1}^{M} \frac{\brac{\mv{u}{n}{j}(t,\mf{g})-\mv{P}{n}{j}}}{\Delta_j} dt < \infty,
\end{equation} 
for all $n \geq 1$. We will use induction starting with $n=1$.
Using~\eqref{eq:v1}, we have

\begin{align*}
\int_{0}^{\tau} \sum_{j=1}^{M} \frac{\brac{\mv{u}{1}{j}(t,\mf{g})-\mv{P}{1}{j}}}{\Delta_j} dt
&= \int_{0}^{\tau} \sum_{j=1}^{M} \brac{\frac{\mv{u}{1}{j}(t,\mf{g})}{\Delta_j}-1} dt\\
&= -\frac{1}{\lambda}\int_{0}^{\tau} \frac{d v_1(t,\mf{g})}{dt} dt\\
&= \frac{1}{\lambda}(v_1(\mf{g})-v_1(\tau,\mf{g})).
\end{align*}
Since the right hand side is bounded by a constant for all $\tau$,
the integral on the left hand side must converge as $\tau \to \infty$.

Now assume that~\eqref{eq:more} holds for all $n \leq L-1$. We have from~\eqref{eq:vn}
and~\eqref{eq:coupling_x}

\begin{align*}
v_L(0,\mf{g})-v_L(\tau,\mf{g}) 
& = -\int_0^{\tau} \frac{dv_L(t,\mf{g})}{dt} dt\\
& = \lambda \int_{0}^{\tau} \brac{\sum_{j=1}^{M} \frac{\mv{u}{L}{j}(t,\mf{g})}{\Delta_j}-\prod_{j=1}^{M}\brac{\mv{u}{L-1}{j}(t,\mf{g})}^{d_j}} dt\\
& = \lambda \int_{0}^{\tau} \sum_{j=1}^{M} \frac{\brac{\mv{u}{L}{j}(t,\mf{g})-\mv{P}{L}{j}}}{\Delta_j} dt\\
& \hspace{2cm} + \lambda \int_{0}^{\tau} \brac{\sum_{j=1}^{M} \frac{\mv{P}{L}{j}}{\Delta_j}-\prod_{j=1}^{M}\brac{\mv{u}{L-1}{j}(t,\mf{g})}^{d_j}} dt\\
& = \lambda \int_{0}^{\tau} \sum_{j=1}^{M} \frac{\brac{\mv{u}{L}{j}(t,\mf{g})-\mv{P}{L}{j}}}{\Delta_j} dt\\
& \hspace{1cm} - \lambda \int_{0}^{\tau} \brac{\prod_{j=1}^{M}\brac{\mv{u}{L-1}{j}(t,\mf{g})}^{d_j}-\prod_{j=1}^{M}\brac{\mv{P}{L-1}{j}}^{d_j}}dt
\end{align*}
By the induction hypothesis, the last integral on the right hand side converges as $\tau \to \infty$.
The left hand side also is uniformly bounded. Hence, the first integral on the left hand side
also must converge as required. \qed

\bibliographystyle{acmtrans-ims}
\bibliography{load_balance}

\end{document}